\documentclass[11pt]{article}

\usepackage[hyphens]{url}
\usepackage[pagebackref,colorlinks]{hyperref}
\usepackage[hyphenbreaks]{breakurl}
\usepackage{amsmath} 
\usepackage{amsthm} 
\usepackage{amssymb}	
\usepackage{graphicx} 
\usepackage{multicol} 
\usepackage{multirow}
\usepackage{color}
\usepackage[dvips,letterpaper,margin=1in,bottom=1in]{geometry}
\usepackage[capitalize,noabbrev]{cleveref}

\usepackage{bm}

\usepackage{ytableau}
\usepackage{threeparttable}

\usepackage[utf8]{inputenc}
\usepackage[english]{babel}

\usepackage[T1]{fontenc}
\AtBeginDocument{%
  \DeclareFontShape{T1}{cmr}{m}{scit}{<->ssub*cmr/m/sc}{}%
}

\usepackage{diagbox}
\usepackage{mathtools}

\newtheorem{theorem}{Theorem}[section]

\newtheorem{lemma}[theorem]{Lemma}

\newtheorem{fact}[theorem]{Fact}

\newtheorem{definition}[theorem]{Definition}

\newcommand{\braket}[2]{\langle #1 | #2 \rangle}

\DeclarePairedDelimiter\rbra{\lparen}{\rparen}
\DeclarePairedDelimiter\sbra{\lbrack}{\rbrack}
\DeclarePairedDelimiter\cbra{\{}{\}}
\DeclarePairedDelimiter\abs{\lvert}{\rvert}
\DeclarePairedDelimiter\Abs{\lVert}{\rVert}
\DeclarePairedDelimiter\ceil{\lceil}{\rceil}

\DeclarePairedDelimiter\ket{\lvert}{\rangle}
\DeclarePairedDelimiter\bra{\langle}{\rvert}
\DeclarePairedDelimiter\ave{\langle}{\rangle}

\DeclareMathOperator*{\E}{\mathbb{E}}
\DeclareMathOperator*{\Var}{\mathbf{Var}}

\newcommand{\C} {\mathbb{C}}

\newcommand{\tr} {\operatorname{tr}}

\newcommand{\diag} {\operatorname{diag}}

\newcommand{\supp} {\operatorname{supp}}

\newcommand{\Real} {\operatorname{Re}}

\newcommand{\kett}[1]{|#1\rangle\!\rangle}

\newcommand{\ketbra}[2]{\ensuremath{\ket{#1}\!\bra{#2}}}

\usepackage{latexsym}
\usepackage{CJK}

\usepackage{enumerate}

\usepackage{algorithm}
\usepackage{algpseudocode}

\usepackage{stmaryrd}
\usepackage{tabularx}
\usepackage{booktabs}
\usepackage{adjustbox}

\newcommand{\footremember}[2]{%
    \footnote{#2}
    \newcounter{#1}
    \setcounter{#1}{\value{footnote}}%
}

\usepackage{tikz}
\usetikzlibrary{quantikz2}

\begin{document}

\title{Nearly Sample-Optimal Estimators for Quantum R{\'e}nyi
and Tsallis Entropies}
\author{Kean Chen\footremember{1}{\url{keanchen.gan@gmail.com}} and Qisheng Wang\footremember{2}{\url{QishengWang1994@gmail.com}}}
\date{}

\maketitle

\begin{abstract}
In this paper, we provide estimators for quantum R\'enyi and Tsallis entropies with nearly optimal sample complexity. 
Specifically, for order $\alpha$, dimension $d$, and additive error $\varepsilon$, 
\begin{itemize}
    \item For $0 < \alpha < 1$, the sample complexity is $O(d^{1+1/\alpha}/\varepsilon^{1/\alpha} + d^{1/\alpha-1}/\varepsilon^{2})$ for R\'enyi entropy and $O(d^{1+1/\alpha}/\varepsilon^{1/\alpha} + d^{2-2\alpha}/\varepsilon^2)$ for Tsallis entropy.
    In particular, for $0 < \alpha \leq 1/2$, the sample complexity for both entropies is $O(d^{1+1/\alpha}/\varepsilon^{1/\alpha})$. 
    \item For non-integer $\alpha > 1$, the sample complexity is $O(d^2/\varepsilon^{1/\alpha} + d^{1-1/\alpha}/\varepsilon^2)$ for R\'enyi entropy.
\end{itemize}
Our upper bounds improve the quantum R\'enyi entropy estimators due to \hyperlink{cite.AISW20}{Acharya, Issa, Shende, and Wagner (2017)} and the quantum Tsallis entropy estimators due to \hyperlink{cite.CLW26}{Chen, Liu, and Wang (2026)}, and match the lower bounds recently established by \hyperlink{cite.Wan26}{Wang (2026)}. 
\end{abstract}

\newpage
\tableofcontents
\newpage

\section{Introduction}

Entropy estimation is a fundamental task in information theory and statistics. 
The estimation of Shannon entropy \cite{Sha48a,Sha48b} has attracted a lot of attention in the literature \cite{Pan03,Pan04,BDKR05,Val11,VV11a,VV11b,VV17,JVHW15,JVHW17,WY16}. 
In particular, for a $d$-dimensional discrete distribution $P$, estimating its Shannon entropy, 
\[
\mathrm{H}\rbra{P} = - \sum_{i=1}^d p_i \log\rbra*{p_i},
\]
to within additive error $\varepsilon$ is known to have sample complexity $\Theta\rbra{\frac{d}{\varepsilon\log\rbra{d}}+\frac{\log^2\rbra{d}}{\varepsilon^2}}$ \cite{JVHW15,WY16}. 
As generalizations of Shannon entropy, (near-)optimal estimators for R\'enyi entropy \cite{Ren61} and Tsallis entropy \cite{Tsa88},
\[
\mathrm{H}_{\alpha}^{\mathrm{R}}\rbra{P} = \frac{1}{1-\alpha} \log\rbra*{\sum_{i=1}^d p_i^\alpha}, \qquad \mathrm{H}_{\alpha}^{\mathrm{T}}\rbra{P} = \frac{1}{1-\alpha}\rbra*{\sum_{i=1}^d p_i^\alpha - 1},
\]
have been proposed in \cite{AOST17} and \cite{JVHW15,JVHW17}, respectively. 

In the quantum world, the estimation of von Neumann entropy \cite{vN27,vN32},
\[
\mathrm{S}\rbra{\rho} = -\tr\rbra*{\rho\log\rbra{\rho}},
\]
has also got extensively studied \cite{BMW16,AISW20,WZ25,GW26}.
A nearly sample-optimal estimator for von Neumann entropy was known in \cite{BMW16} with sample complexity $O\rbra{\frac{d^2}{\varepsilon} + \frac{\log^2\rbra{d}}{\varepsilon^2}}$ (cf.\ \cite[Theorem 1.27]{OW17}), where an almost matching lower bound of $\widetilde{\Omega}\rbra{\frac{d^2}{\varepsilon}+\frac{1}{\varepsilon^2}}$ was recently established in \cite{Wan26}.\footnote{A lower bound of $\Omega\rbra{d^{2-\gamma}}$ was also recently established in \cite{FOW26} for any constant $\gamma > 0$.} 
However, the sample complexities of estimating the quantum R\'enyi entropy and quantum Tsallis entropy,
\begin{equation} \label{eq:global-functional}
\mathrm{S}_{\alpha}^{\mathrm{R}}\rbra{\rho} = \frac{1}{1-\alpha} \log\rbra*{\mathrm{F}_\alpha\rbra{\rho}}, \qquad \mathrm{S}_\alpha^{\mathrm{T}}\rbra{\rho} = \frac{1}{1-\alpha} \rbra*{\mathrm{F}_\alpha\rbra{\rho} - 1}, \qquad \textup{ where } \mathrm{F}_\alpha\rbra{\rho} = \tr\rbra*{\rho^\alpha},
\end{equation}
have not been fully settled. 
\begin{itemize}
    \item \textbf{R\'enyi entropy.}
    \begin{itemize}
        \item For $0 < \alpha < 1$, the upper bound is $O\rbra{\frac{d^{2/\alpha}}{\varepsilon^{2/\alpha}}}$ \cite{AISW20}, whereas the lower bound $\widetilde{\Omega}\rbra{\frac{d^{1+1/\alpha}}{\varepsilon^{1/\alpha}}+\frac{d^{1/\alpha-1}}{\varepsilon^2}}$ \cite{Wan26}. 
        \item For non-integer $\alpha > 1$, the upper bound is $O\rbra{\frac{d^2}{\varepsilon^2}}$ \cite{AISW20}, whereas the lower bound is $\widetilde{\Omega}\rbra{\frac{d^2}{\varepsilon^{1/\alpha}}+\frac{d^{1-1/\alpha}}{\varepsilon^2}}$ \cite{Wan26}. 
        \item For integer $\alpha \geq 2$, the sample complexity is known to be $\Theta\rbra{\frac{d^{2-2/\alpha}}{\varepsilon^{2/\alpha}}+\frac{d^{1-1/\alpha}}{\varepsilon^2}}$ \cite{AISW20}. 
    \end{itemize}
    \item \textbf{Tsallis entropy.}
    \begin{itemize}
        \item For $0 < \alpha < 1$, the upper bound is $O\rbra{\frac{d^{2/\alpha}}{\varepsilon^{2/\alpha}}}$ \cite{CLW26}, whereas the lower bound is $\widetilde{\Omega}\rbra{\frac{d^{1+1/\alpha}}{\varepsilon^{1/\alpha}}+\frac{d^{2-2\alpha}}{\varepsilon^2}}$ \cite{Wan26}. 
        \item For non-integer $\alpha > 1$, the sample complexity is shown to be $\widetilde{\Theta}\rbra{\frac{1}{\varepsilon^{\max\cbra{2/\rbra{\alpha-1}, 2}}}}$ \cite{CW25,Wan26}. 
        \item For integer $\alpha \geq 2$, the sample complexity is known to be $\Theta\rbra{\frac{1}{\varepsilon^2}}$, where the upper bound is by the generalized SWAP test \cite{BCWdW01,EAO+02} and the lower bound was recently established in \cite{CWLY23,GHYZ26,CWYZ26}.
    \end{itemize}
\end{itemize}

In this paper, we close all the remaining gaps between the sample complexity upper and lower bounds by providing near-optimal estimators for (i) quantum R\'enyi entropy for all non-integer order $\alpha$ and (ii) quantum Tsallis entropy for $0 < \alpha < 1$, leaving only room for polylogarithmic improvements. 

\begin{theorem}[Nearly sample-optimal estimators for quantum R\'enyi and Tsallis entropies] \label{thm:main}
    Given sample access to an unknown $d$-dimensional quantum state $\rho$ and the additive error $\varepsilon$, 
    \begin{itemize}
        \item For $0 < \alpha < 1$, we can estimate the $\alpha$-R\'enyi entropy $\mathrm{S}_\alpha^{\mathrm{R}}\rbra{\rho}$ and Tsallis entropy $\mathrm{S}_\alpha^{\mathrm{T}}\rbra{\rho}$ using 
        \[
        O\rbra*{\frac{d^{1+1/\alpha}}{\varepsilon^{1/\alpha}}+\frac{d^{1/\alpha-1}}{\varepsilon^2}}, \qquad O\rbra*{\frac{d^{1+1/\alpha}}{\varepsilon^{1/\alpha}}+\frac{d^{2-2\alpha}}{\varepsilon^2}},
        \]
        samples of $\rho$, respectively. 
        In particular, for $0 < \alpha \leq 1/2$, both estimators have sample complexity 
        \[
        O\rbra*{\frac{d^{1+1/\alpha}}{\varepsilon^{1/\alpha}}}.
        \]
        \item For non-integer $\alpha > 1$, we can estimate the $\alpha$-R\'enyi entropy $\mathrm{S}_\alpha^{\mathrm{R}}\rbra{\rho}$ using 
        \[
        O\rbra*{\frac{d^2}{\varepsilon^{1/\alpha}}+\frac{d^{1-1/\alpha}}{\varepsilon^2}}
        \]
        samples of $\rho$. 
    \end{itemize}
\end{theorem}

All the sample complexity upper bounds in \cref{thm:main} match the lower bounds in \cite{Wan26} up to polylogarithmic factors in $d$ and $1/\varepsilon$.
We summarize the sample complexities of estimating the quantum R\'enyi and Tsallis entropies in \cref{tab1}. 

\begin{table}[t]
    \centering

    \begin{threeparttable}
    \caption{Sample complexity of estimating quantum R\'enyi and Tsallis entropies.}
    \label{tab1}\vspace{2mm}
    \begin{tabular}{ccc}
        \toprule
         Order & R\'enyi Entropy & Tsallis Entropy \\
         \midrule
         $0 < \alpha \leq 1/2$ & \multicolumn{2}{c}{\begin{tabular}{c}$\widetilde{\Theta}\rbra*{\dfrac{d^{1+1/\alpha}}{\varepsilon^{1/\alpha}}}$ \\ \addlinespace[2pt] This Work$^{\uparrow}$ and \cite{Wan26}$^{\downarrow}$ \end{tabular}} \\
         \midrule
         $1/2 < \alpha < 1$ & \begin{tabular}{c}
            $\widetilde{\Theta}\rbra*{\dfrac{d^{1+1/\alpha}}{\varepsilon^{1/\alpha}}+\dfrac{d^{1/\alpha-1}}{\varepsilon^2}}$ \\ \addlinespace[2pt] This Work$^\uparrow$ and \cite{Wan26}$^\downarrow$
         \end{tabular} & \begin{tabular}{c}
            $\widetilde{\Theta}\rbra*{\dfrac{d^{1+1/\alpha}}{\varepsilon^{1/\alpha}}+\dfrac{d^{2-2\alpha}}{\varepsilon^2}}$ \\ \addlinespace[2pt] This Work$^\uparrow$ and \cite{Wan26}$^\downarrow$
         \end{tabular} \\
         \midrule
         \begin{tabular}{c}$\alpha = 1$ \\ (von Neumann) \end{tabular} & \multicolumn{2}{c}{\begin{tabular}{c}
            $\widetilde{\Theta}\rbra*{\dfrac{d^2}{\varepsilon} + \dfrac{1}{\varepsilon^2}}$ \\ \addlinespace[2pt]
            \cite{BMW16,OW17}$^\uparrow$ and \cite{Wan26}$^\downarrow$
         \end{tabular}} \\
         \midrule
         $1 < \alpha < 2$ & \multirow{2}{*}{\begin{tabular}{c}$\widetilde{\Theta}\rbra*{\dfrac{d^2}{\varepsilon^{1/\alpha}}+\dfrac{d^{1-1/\alpha}}{\varepsilon^2}}$ \\ \addlinespace[2pt]
             This Work$^\uparrow$ and \cite{Wan26}$^\downarrow$
             \end{tabular}} 
             & \begin{tabular}{c} $\widetilde{\Theta}\rbra*{\dfrac{1}{\varepsilon^{2/\rbra{\alpha-1}}}}$ \\ \addlinespace[2pt]
              \cite{CW25}$^\uparrow$ and \cite{Wan26}$^\downarrow$ \end{tabular} \\
         \cmidrule{1-1}\cmidrule{3-3}
         Non-integer $\alpha > 2$ & & \begin{tabular}{c} $\widetilde{\Theta}\rbra*{\dfrac{1}{\varepsilon^{2}}}$ \\ \addlinespace[2pt]
              \cite{CW25}$^{\uparrow\downarrow}$ \end{tabular} \\
         \midrule
         Integer $\alpha \geq 2$ & \begin{tabular}{c} ${\Theta}\rbra*{\dfrac{d^{2-2/\alpha}}{\varepsilon^{2/\alpha}}+\dfrac{d^{1-1/\alpha}}{\varepsilon^2}}$ \\ \addlinespace[2pt]
              \cite{AISW20}$^{\uparrow\downarrow}$ \end{tabular} & \begin{tabular}{c} $\Theta\rbra*{\dfrac{1}{\varepsilon^{2}}}$ \\ \addlinespace[2pt]
              \cite{BCWdW01,EAO+02}$^\uparrow$ \\
              \cite{CWLY23,GHYZ26,CWYZ26}$^\downarrow$ \end{tabular} \\
        \bottomrule
    \end{tabular}
    \begin{tablenotes}
        \item[$\uparrow$] References for (near-)optimal upper bounds.
        \item[$\downarrow$] References for (near-)optimal lower bounds.
    \end{tablenotes}
    \end{threeparttable}
\end{table}

\subsection{Techniques}

Our estimators are different from the known ones in the literature \cite{AISW20,LW26,CW25,CLW26}. 
In comparison, the previous estimators build on weak Schur sampling \cite{CHW07} while our estimators build on specific quantum state tomography with new information-theoretic inequalities. 
Our estimators for $0 < \alpha < 1$ and $\alpha > 1$ are different. 
So we introduce them below separately. 

\paragraph{Some notations and facts.}
We use $\operatorname{Clip}_{[u, v]}\rbra{x} = \min\cbra{\max\cbra{x, u}, v}$ to ensure an estimate in a valid range, given the bounds
\begin{equation}
    \label{eq:global-physical-range}
    \mathrm{F}_\alpha\rbra{\rho} = \tr\rbra{\rho^\alpha} \in \begin{cases}
        \sbra{1, d^{1-\alpha}} & 0 < \alpha < 1, \\
        \sbra{d^{1-\alpha}, 1} & \alpha > 1.
    \end{cases}
\end{equation}
To estimate the R\'enyi entropy, we estimate $\mathrm{F}_\alpha\rbra{\rho}$ to a relative error, using the following fact.
\begin{fact} \label{lem:global-log-conversion}
    For $F, \widehat{F}, \varepsilon > 0$, if $\abs{\widehat{F} - F} \leq \rbra{1 - e^{-\varepsilon}} F$, then $\abs{\log\rbra{F} - \log\rbra{\widehat{F}}} \leq \varepsilon$.
\end{fact}

\paragraph{The case of $0 < \alpha < 1$.}
Our estimator is based on the new inequality involving the Bures $\chi^2$-divergence (see \cref{lem:one-step-bias}): For $d$-dimensional density operators $\rho$ and $\sigma$ with $\sigma \succeq \frac{d}{n}I$, 
\[
0 \leq \rbra*{1-\alpha}\tr\rbra*{\sigma^\alpha} +\alpha\tr\rbra*{\rho\sigma^{\alpha-1}} -\tr\rbra{\rho^\alpha} \leq \rbra*{\frac{d}{n}}^{\alpha-1} \mathrm{D}_{\chi^2}\rbra{\rho\,\|\,\sigma}.
\]
This suggests that we take $\rbra{1-\alpha}\tr\rbra*{\sigma^\alpha} +\alpha\tr\rbra{\rho\sigma^{\alpha-1}}$ as an estimate of $\mathrm{F}_\alpha\rbra{\rho} = \tr\rbra{\rho^\alpha}$.
This estimate is good when the Bures $\chi^2$-divergence of $\rho$ to $\sigma$ is small. 
To this end, we use quantum state tomography in \cite{PSTW26}, which gives an estimate $\widehat{\rho}$ from $n$ samples of $\rho$; further analysis shows that taking $\sigma = \widehat{\rho} + \frac{2d}{n}I$ yields a small enough Bures $\chi^2$-divergence (see \cref{lem:regularized-chi2}):
\[
\mathrm{D}_{\chi^2}\rbra{\rho\,\|\,\sigma} \leq O\rbra*{\frac{d^2}{n}}.
\]
With these, the estimator then consists of two components: (i) estimate $\rbra{1-\alpha}\tr\rbra{\sigma^\alpha}$, which can be calculated directly from the known $\sigma$; and (ii) estimate $\alpha\tr\rbra{\rho\sigma^{\alpha-1}}$, which can be done by measuring the observable $\sigma^{\alpha-1}$ on $\rho$.
Let $m$ be the number of samples of $\rho$ used in step (ii). 
The total sample complexity is then simply $n + m$. 
By choosing $n$ and $m$ appropriately with detailed error analysis, we can yield the sample complexity upper bound presented in \cref{thm:main} for $0 < \alpha < 1$. 

\paragraph{The case of $\alpha > 1$.}
In this case, our estimator is based on a new analysis of Hayashi's pure state estimation algorithm~\cite{Hay98} equipped with random purification techniques~\cite{tang2025conjugate}. 
Suppose $\rho$ is an unknown state and $\ket{\psi}$ is a random purification of $\rho$.
Now, we consider a batch of $s$ samples of $\ket{\psi}$.
Let $\ket{\widehat{\psi}}$ be the output of Hayashi's state estimation algorithm applied on $\ket{\psi}^{\otimes s}$. 
We define a one-batch estimate of $F\coloneqq \tr\rbra*{\rho^\alpha}$ with batch size $s$ as:
\[Y_s\coloneqq \tr(\widehat\rho^\alpha),\quad \textup{where } \widehat{\rho}=\tr_{\mathrm{anc}}\rbra{\ket{\widehat{\psi}}\!\bra{\widehat{\psi}}}.\]
Central to our proof is the following expansion for the bias (see \cref{gt:prop:one-batch}):
\begin{equation*}
 \E \sbra{Y_s} - F
 =\sum_{j=1}^k c_j\mu_j(s)+R(s),
 \end{equation*}
where $c_j$ are $s$-independent coefficients, $R(s)$ is a small remainder with $|R(s)|=O(F\cdot (\frac{d^2}{s})^\alpha)$ and $\mu_j(s)$ is the $j$-th moment of beta distribution parameterized by $s$.
Then, we repeat the above process for many times for different batch sizes $s_1,\ldots,s_k$.
The final estimator for $F$ is defined as $\widehat{F}\coloneqq \sum_{i=1}^k a_i Y_{s_i}$, where $a_1,\ldots,a_k$ are carefully chosen coefficients such that $\sum_{i=1}^k a_i \mu_{j}(s_i)=0$.
This means all the $i$-th order term for $i\leq k$ can be canceled and the bias of $\widehat{F}$ is just a linear combination of the remainders: $\widehat{F}-F=\sum_i a_i R(s_i)$.
As $\sum_i |a_i|$ can also be constantly bounded, we can show that the bias $|\widehat{F}-F|$.
With a similar but more complicated analysis, the variance of $\widehat{F}$ can also be well bounded (see \cref{gt:prop:combined}).
These produce a good relative error estimate of $F$, and thus a good additive estimate of the R\'enyi entropy.

\subsection{Related work}

The estimation of quantum R\'enyi and Tsallis entropies has also been studied in other settings. 
In the incoherent setting, the estimation of quantum R\'enyi entropy of integer order was studied in \cite{PTTW26}. 
Given purified quantum query access, quantum R\'enyi entropy estimation was studied in \cite{SH21,WGL+24,WZL24} and quantum Tsallis entropy estimation was studied in \cite{LW26,Wan25};
for the special case of order $\alpha = 1$ (von Neumann entropy), von Neumann entropy estimation was studied in \cite{GL20,GHS21,WGL+24}.

\section{R\'enyi and Tsallis Entropy with \texorpdfstring{$0 < \alpha < 1$}{0 < α < 1}}

Fix $0<\alpha<1$.  
Both of our estimators for $\alpha$-R\'enyi and $\alpha$-Tsallis entropies first estimate $\mathrm{F}_\alpha\rbra{\rho}$ in the same framework, though with different parameter choices. 
The Tsallis entropy
estimator requires additive error, whereas the R\'enyi estimator requires
relative error.
The main results of this section are stated as follows. 

\begin{theorem} [Quantum R\'enyi and Tsallis entropy estimation for $0 < \alpha < 1$] \label{thm:leq1}
    For $0 < \alpha < 1$, given sample access to an unknown $d$-dimensional quantum state $\rho$ and the additive error $\varepsilon$, we can estimate the $\alpha$-R\'enyi entropy $\mathrm{S}_\alpha^{\mathrm{R}}\rbra{\rho}$ and Tsallis entropy $\mathrm{S}_\alpha^{\mathrm{T}}\rbra{\rho}$ using 
        \[
        O\rbra*{\frac{d^{1+1/\alpha}}{\varepsilon^{1/\alpha}}+\frac{d^{1/\alpha-1}}{\varepsilon^2}}, \qquad O\rbra*{\frac{d^{1+1/\alpha}}{\varepsilon^{1/\alpha}}+\frac{d^{2-2\alpha}}{\varepsilon^2}},
        \]
        samples of $\rho$, respectively. 
\end{theorem}

\subsection{The estimator}

Our estimator is based on the quantum state tomography provided in \cite{PSTW26} which ensures a good precision in Bures $\chi^2$-divergence. 

\begin{lemma}[Bures $\chi^2$-divergence tomography, {\cite[Corollary 4.10 and Lemma 6.9]{PSTW26}}] \label{lemma:bures-chi2-tomo}
    There is a quantum algorithm $\mathsf{TomoBures\chi^2}\rbra{\rho, d, n}$ that, on input $n$ samples of an unknown $d$-dimensional quantum state $\rho$, outputs a $d$-dimensional Hermitian matrix 
    \begin{equation} \label{eq:new-hatrho-spectrum}
    \widehat{\rho} = \sum_{i=1}^d \widehat{\lambda}_i \ketbra{v_i}{v_i} \succeq -\frac{d}{n} I,
    \qquad
    \widehat{\lambda}_1 \geq \cdots \geq \widehat{\lambda}_d,
    \end{equation}
    with $\tr\rbra{\widehat{\rho}} = 1$ such that with probability at least $0.99$, the following holds:
    \begin{enumerate}
        \item Event $E$: For every unit vector $\ket{w} \in \mathbb{C}^d$, 
        \[
        \abs{\bra{w} \rbra{\widehat{\rho} - \rho} \ket{w}} \leq C \sqrt{\frac{d}{n}\rbra*{\bra{w} \widehat{\rho} \ket{w} + \frac{d}{n}}},
        \]
        where $C > 0$ is a universal constant. 
        \item On event $E$, for every $m \in \sbra{d}$, 
        \[
        \sum_{\substack{i, j \in \sbra{d} \\ \min\cbra{i,j} = m}} \abs*{ \bra{v_i} \rbra*{\widehat{\rho} - \rho} \ket{v_j} }^2 \leq 2C^2 \frac{d}{n} \rbra*{\widehat{\lambda}_m + \frac{d}{n}}.
        \]
    \end{enumerate}
\end{lemma}

Let $n$ be the sample size for tomography and let $m$ be the sample size for bias correction.  The two batches are disjoint.
The formal description of our estimator is given in \cref{alg:powertrace}.

\begin{algorithm}[htbp]
\caption{$\mathsf{PowerTrace}\rbra{\rho,d,\alpha,n,m}$}
\label{alg:powertrace}
\begin{algorithmic}[1]
\Require Independent samples of an unknown $d$-dimensional quantum state $\rho$;
$0<\alpha<1$; positive integers $n,m$.
\Ensure An estimate $\widehat{\mathrm{F}}_\alpha\in[1,d^{1-\alpha}]$ of $\mathrm{F}_\alpha\rbra{\rho} = \tr\rbra{\rho^\alpha}$.
\State $\widehat{\rho} \gets \mathsf{TomoBures\chi^2}\rbra{\rho,d,n}$, using $n$ samples of $\rho$. \Comment{see \cref{lemma:bures-chi2-tomo}}
\State $\sigma\gets\widehat{\rho}+\frac{2d}{n}I$, and decompose
$\sigma=\sum_{i=1}^d s_i\ketbra{v_i}{v_i}$.
\For{$j=1,\ldots,m$}
  \State Measure a sample of $\rho$ in the basis $\cbra*{\ketbra{v_i}{v_i}}_{i=1}^d$ and let $J_j$ be the outcome.
  \State Set $X_j\gets s_{J_j}^{\alpha-1}$.
\EndFor
\State $\widetilde{\mathrm{F}}_\alpha
 \gets
 \rbra{1-\alpha}\tr\rbra{\sigma^\alpha}
 +\frac{\alpha}{m}\sum_{j=1}^m X_j$.
\State \Return $\widehat{\mathrm{F}}_\alpha \gets \operatorname{Clip}_{[1, d^{1-\alpha}]}\rbra{\widetilde{\mathrm{F}}_\alpha}$. 
\end{algorithmic}
\end{algorithm}

\cref{alg:powertrace} gives estimates with different precisions depending on the choice of $n$ and $m$.
Below we collect the error results for our purposes. 

\begin{lemma}
\label{thm:common-estimator}
Fix $0<\alpha<1$.  There are constants
$C_\alpha^{(0)},C_\alpha^{(1)}\geq1$, depending only on $\alpha$, for
which the following holds.  Let $d\geq2$, let $\rho$ be a
$d$-dimensional density operator, and denote $F:=\mathrm{F}_\alpha\rbra{\rho}=\tr\rbra{\rho^\alpha}$.
Let $\widehat{\mathrm{F}}_\alpha$ be the output of $\mathsf{PowerTrace}\rbra{\rho,d,\alpha,n,m}$ defined in \cref{alg:powertrace}, where $n = \ceil{C_\alpha^{(0)}d^{1+1/\alpha}/\varepsilon^{1/\alpha}}$, and $m$ is to be specified below. 
Then, 
\begin{enumerate}[(i)]
 \item If $0<\alpha\leq1/2$, then set $m=\ceil{C_\alpha^{(1)}d^{1/\alpha-1}/\varepsilon^{1/\alpha}}$, which gives $\Pr\sbra{\abs{\widehat{\mathrm{F}}_{\alpha}-F} \leq \varepsilon} \geq 0.98$.
 
 \item If $1/2<\alpha<1$, then set $m=\ceil{C_\alpha^{(1)}d^{2-2\alpha}/\varepsilon^{2}}$,
 which gives $\Pr\sbra{\abs{\widehat{\mathrm{F}}_\alpha-F}\leq\varepsilon} \geq 0.98$.

 \item If $1/2<\alpha<1$, then, set $m=\ceil{C_\alpha^{(1)}d^{1/\alpha-1}/\varepsilon^{2}}$, which gives $\Pr\sbra{\abs{\widehat{\mathrm{F}}_\alpha-F}\leq\varepsilon F} \geq 0.98$.
\end{enumerate}
\end{lemma}

For readability, we postpone the proof of \cref{thm:common-estimator} in \cref{sec:proof-common-estimator}. 

\subsubsection{Tsallis entropy}
 
The formal description of our Tsallis entropy estimator is given in \cref{alg:tsallis}.

\begin{algorithm}[htbp]
\caption{$\mathsf{TsallisEntropy}_{<1}\rbra{\rho,d,\alpha,\varepsilon}$}
\label{alg:tsallis}
\begin{algorithmic}[1]
\Require Independent samples of an unknown $d$-dimensional quantum state $\rho$;
$0<\alpha<1$; $0<\varepsilon\leq1$.
\Ensure An estimate
$\widehat{\mathrm{S}}_\alpha^{\mathrm{T}}
 \in\sbra*{0,\rbra{d^{1-\alpha}-1}/\rbra{1-\alpha}}$
of $\mathrm{S}_\alpha^{\mathrm{T}}\rbra{\rho}$.
\State $\delta\gets\rbra{1-\alpha}\varepsilon$.
\State $n\gets
 \ceil{C_\alpha^{(0)}d^{1+1/\alpha}\delta^{-1/\alpha}}$.
\If{$0<\alpha\leq1/2$}
 \State $m\gets
  \ceil{C_\alpha^{(1)}d^{1/\alpha-1}\delta^{-1/\alpha}}$.
\Else
 \State $m\gets
  \ceil{C_\alpha^{(1)}d^{2-2\alpha}\delta^{-2}}$.
\EndIf
\State $\widehat{\mathrm{F}}_\alpha\gets
 \mathsf{PowerTrace}\rbra{\rho,d,\alpha,n,m}$.
\State \Return
 $\widehat{\mathrm{S}}_\alpha^{\mathrm{T}}
 \gets\dfrac{\widehat{\mathrm{F}}_\alpha-1}{1-\alpha}$.
\end{algorithmic}
\end{algorithm}

\begin{proof}[Proof of \cref{thm:leq1} (Tsallis)]
Let $F:=\mathrm{F}_\alpha\rbra{\rho}$.  The choices of $n$ and $m$ in
\cref{alg:tsallis} are those prescribed by
\cref{thm:common-estimator} with $\varepsilon \coloneqq \delta$.
Since $0<\delta\leq1$, parts~(i) and~(ii) of that lemma, for
$0<\alpha\leq1/2$ and $1/2<\alpha<1$, respectively, give $\Pr\sbra{\abs{\widehat{\mathrm{F}}_\alpha-F}\leq\delta}
 \geq0.98$.
On this event, the definition of $\mathrm{S}_\alpha^{\mathrm{T}}$ yields
\begin{align*}
 \abs*{\widehat{\mathrm{S}}_\alpha^{\mathrm{T}}
 -\mathrm{S}_\alpha^{\mathrm{T}}\rbra{\rho}}
 =\frac{1}{1-\alpha}
 \abs*{\widehat{\mathrm{F}}_\alpha-F}\leq\frac{\delta}{1-\alpha}
 =\varepsilon.
\end{align*}
Thus the success probability is at least $0.98>2/3$.

The total number of samples of $\rho$ is
\[
 n+m=O_\alpha\rbra*{
 \frac{d^{1+1/\alpha}}{\varepsilon^{1/\alpha}}
 +\frac{d^{2-2\alpha}}{\varepsilon^{2}} }.
\]
\end{proof}

\subsubsection{R\'enyi entropy}

The formal description of our R\'enyi entropy estimator is given in \cref{alg:renyi}.

\begin{algorithm}[htbp]
\caption{$\mathsf{RenyiEntropy}_{<1}\rbra{\rho,d,\alpha,\varepsilon}$}
\label{alg:renyi}
\begin{algorithmic}[1]
\Require Independent samples of an unknown $d$-dimensional quantum state $\rho$;
$0<\alpha<1$; $0<\varepsilon\leq1$.
\Ensure An estimate
$\widehat{\mathrm{S}}_\alpha^{\mathrm{R}}
 \in\sbra*{0,\log d}$
of $\mathrm{S}_\alpha^{\mathrm{R}}\rbra{\rho}$.
\State $\delta\gets1-e^{-\rbra{1-\alpha}\varepsilon}$.
\State $n\gets
 \ceil{C_\alpha^{(0)}d^{1+1/\alpha}\delta^{-1/\alpha}}$.
\If{$0<\alpha\leq1/2$}
 \State $m\gets
  \ceil{C_\alpha^{(1)}d^{1/\alpha-1}\delta^{-1/\alpha}}$.
\Else
 \State $m\gets
  \ceil{C_\alpha^{(1)}d^{1/\alpha-1}\delta^{-2}}$.
\EndIf
\State $\widehat{\mathrm{F}}_\alpha\gets
 \mathsf{PowerTrace}\rbra{\rho,d,\alpha,n,m}$.
\State \Return
 $\widehat{\mathrm{S}}_\alpha^{\mathrm{R}}
 \gets\dfrac{\log\rbra{\widehat{\mathrm{F}}_\alpha}}{1-\alpha}$.
\end{algorithmic}
\end{algorithm}

\begin{proof}[Proof of \cref{thm:leq1} (R\'enyi)]
Let $F \coloneqq \mathrm{F}_\alpha\rbra{\rho}$ and
$\delta \coloneqq 1-e^{-\rbra{1-\alpha}\varepsilon} = \Theta_\alpha\rbra{\varepsilon}$.  The choices of $n$ and
$m$ in \cref{alg:renyi} are those prescribed by
\cref{thm:common-estimator} with $\varepsilon \coloneqq \delta$.
Since $0<\delta<1$, if $0<\alpha\leq1/2$, part~(i) of that lemma gives
$\abs{\widehat{\mathrm{F}}_\alpha-F}\leq\delta$ with probability at
least $0.98$.  Since $F\geq1$ by \cref{eq:global-physical-range}, this
implies $\abs{\widehat{\mathrm{F}}_\alpha-F}\leq\delta F$.  If
$1/2<\alpha<1$, part~(iii) gives the same relative bound directly.
Thus, in either regime, $\Pr\sbra{\abs{\widehat{\mathrm{F}}_\alpha-F}\leq\delta F}
 \geq0.98$.
By \cref{lem:global-log-conversion}, we have
\[
\Pr\sbra*{\abs*{\widehat{\mathrm{S}}_\alpha^{\mathrm{R}}
 -\mathrm{S}_\alpha^{\mathrm{R}}\rbra{\rho}}
 \leq\varepsilon} \geq 0.98.
\]

The total number of samples of $\rho$ is
\[
 n+m=O_\alpha\rbra*{
 \frac{d^{1+1/\alpha}}{\varepsilon^{1/\alpha}}
 + \frac{d^{1/\alpha-1}}{\varepsilon^{2}} }.
\]
\end{proof}

\subsection{Technical lemmas}
\label{sec:part-one-technical}

Our error analysis involves the Bures $\chi^2$-divergence. 

\begin{definition}[Bures $\chi^2$-divergence]\label{def:chi2}
Let $\rho, \sigma \in \mathbb{C}^{d \times d}$ be positive semidefinite operators with $\sigma = \sum_{i=1}^d \lambda_i \ketbra{u_i}{u_i}$.
If $\supp\rbra{\rho} \subseteq \supp\rbra{\sigma}$, then the Bures $\chi^2$-divergence of $\rho$ to $\sigma$ is defined by
\[
\mathrm{D}_{\chi^2}\rbra{\rho \,\|\, \sigma}
= \sum_{\substack{i,j\in\sbra{d}\\\lambda_i+\lambda_j>0}}
\frac{2}{\lambda_i + \lambda_j}
\abs*{ \bra{u_i} \rbra{\rho - \sigma} \ket{u_j} }^2.
\]
If $\supp\rbra{\rho} \not\subseteq \supp\rbra{\sigma}$, define $\mathrm{D}_{\chi^2}\rbra{\rho \,\|\, \sigma} = +\infty$. 
\end{definition}

\subsubsection{Bures \texorpdfstring{$\chi^2$}{chi-square}-divergence bounds for tomography}

\begin{lemma}\label{lem:regularized-chi2}
    Let $\widehat{\rho}$ be the output matrix of $\mathsf{TomoBures\chi^2}\rbra{\rho, d, n}$ and set $\sigma = \widehat{\rho} + \frac{2d}{n}I$. 
    Then, $\sigma \succeq \frac{d}{n} I$ and $\tr\rbra{\sigma} = 1 + \frac{2d^2}{n}$.
    Moreover, on event $E$ in \cref{lemma:bures-chi2-tomo}, 
    \begin{equation}\label{eq:part-one-chi2}
    \mathrm{D}_{\chi^2}\rbra{\rho \,\|\, \sigma} \leq 8\rbra{C^2+1} \frac{d^2}{n},
    \end{equation}
    where $C$ is the constant in \cref{lemma:bures-chi2-tomo}. 
\end{lemma}

\begin{proof}
By \cref{eq:new-hatrho-spectrum} and
$\tr\rbra{\widehat{\rho}}=1$, we have
\[
 \sigma\succeq-\frac{d}{n}I+\frac{2d}{n}I
 =\frac{d}{n}I,
 \qquad
 \tr\rbra{\sigma}
 =\tr\rbra{\widehat{\rho}}+\frac{2d}{n}\tr\rbra{I}
 =1+\frac{2d^2}{n}.
\]
In particular, $\sigma\succ0$, so the divergence $\mathrm{D}_{\chi^2}\rbra{\rho \,\|\, \sigma}$ is
finite.
Use the eigenbasis in \cref{eq:new-hatrho-spectrum},
\[
\sigma = \sum_{i=1}^d s_i \ketbra{v_i}{v_i}, \qquad s_i = \widehat{\lambda}_i + \frac{2d}{n} \geq \frac{d}{n}.
\]
Let $\Delta = \rho - \widehat{\rho}$. 
Then, $\rho - \sigma = \Delta - \frac{2d}{n}I$, which gives
\begin{align*}
\mathrm{D}_{\chi^2}\rbra{\rho \,\|\, \sigma} & = \sum_{i=1}^d \sum_{j=1}^d \frac{2}{s_i + s_j} \abs*{ \bra{v_i} \rbra*{\Delta - \frac{2d}{n}I} \ket{v_j} }^2 \\
& \leq \sum_{i=1}^d \sum_{j=1}^d \frac{2}{s_i + s_j} \rbra*{ 2\abs*{\bra{v_i} \Delta \ket{v_j}}^2 + 2\abs*{\frac{2d}{n}\braket{v_i}{v_j}}^2 } \\
& = 4 \sum_{i=1}^d \sum_{j=1}^d \frac{1}{s_i+s_j} \abs*{\bra{v_i} \Delta \ket{v_j}}^2 +\frac{8d^2}{n^2} \sum_{i=1}^d \frac{1}{s_i}. 
\end{align*}
Since $s_i\ge d/n$, the second term is then bounded by
\[
\frac{8 d^2}{n^2} \sum_{i=1}^d \frac{1}{s_i} \leq \frac{8d^2}{n}.
\]
For the first term, enumerating $\rbra{i, j}$ with $\min\cbra{i, j} = m$ with \cref{lemma:bures-chi2-tomo} gives
\begin{align*}
    4 \sum_{i=1}^d \sum_{j=1}^d \frac{1}{s_i+s_j} \abs*{\bra{v_i} \Delta \ket{v_j}}^2
    & = 4 \sum_{m=1}^d \sum_{\substack{i, j \in \sbra{d} \\ \min\cbra{i, j} = m}} \frac{1}{s_i+s_j} \abs*{\bra{v_i} \Delta \ket{v_j}}^2 \\
    & \leq 4 \sum_{m=1}^d \frac{1}{s_m} \sum_{\substack{i, j \in \sbra{d} \\ \min\cbra{i, j} = m}} \abs*{\bra{v_i} \Delta \ket{v_j}}^2 \\
    & \leq 4 \sum_{m=1}^d \frac{1}{s_m} \cdot 2C^2 \frac{d}{n} \rbra*{\widehat{\lambda}_m + \frac{d}{n}} \\
    & = 8C^2 \frac{d}{n} \sum_{m=1}^d \frac{\widehat{\lambda}_m + d/n}{\widehat{\lambda}_m + {2d}/{n}} \\
    & \leq 8C^2 \frac{d^2}{n}.
\end{align*}
Therefore, 
\[
\mathrm{D}_{\chi^2}\rbra{\rho \,\|\, \sigma} \leq 8\rbra{C^2 + 1} \frac{d^2}{n}. 
\]
\end{proof}

\subsubsection{Bias bounds by Bures \texorpdfstring{$\chi^2$}{chi-square}-divergence}

\begin{lemma}\label{lem:one-step-bias}
Let $0<\alpha<1$, let $n$ be a positive integer, let $\rho$ be a
$d$-dimensional density operator, and let
$\sigma\succeq\frac{d}{n}I$.  Then
\begin{align}
0
\leq
\rbra*{1-\alpha}\tr\rbra*{\sigma^\alpha}
+\alpha\tr\rbra*{\rho\sigma^{\alpha-1}}
-\tr\rbra{\rho^\alpha}
\leq
\rbra*{\frac{d}{n}}^{\alpha-1}
\mathrm{D}_{\chi^2}\rbra{\rho\,\|\,\sigma}.
\label{eq:one-step-bias}
\end{align}
\end{lemma}

To prove \cref{lem:one-step-bias}, we need the following results in matrix analysis. 

\begin{lemma}[{\cite[Theorem V.3.3]{Bha97}}] \label{lemma:dfah}
    Let $f \in C^1\rbra{I}$, where $I$ is an interval, and let
    $A = U \Gamma U^\dag$ be Hermitian, where $U$ is unitary,
    $\Gamma = \diag\rbra{\lambda_1, \lambda_2, \dots, \lambda_d}$,
    and $\lambda_i\in I$ for every $i\in\sbra{d}$.
    If $H$ is Hermitian, then
    \[
    \mathrm{D} f\rbra{A} \sbra{H} \coloneqq \left.\frac{\mathrm{d}}{\mathrm{d}t} f\rbra{A + tH}\right|_{t=0} = \lim_{t \to 0} \frac{f\rbra{A+tH} - f\rbra{A}}{t} = U \rbra*{ f^{[1]}\rbra{\Gamma} \odot U^\dag H U } U^\dag,
    \]
    where $A \odot B$ is the Hadamard element-wise product and $f^{[1]}\rbra{\Gamma}$ is a matrix defined by $\rbra{f^{[1]}\rbra{\Gamma}}_{ij} = f^{[1]}\rbra{\lambda_i, \lambda_j}$, with
    \[
    f^{[1]}\rbra{x, y} = \begin{cases}
        \dfrac{f\rbra{x} - f\rbra{y}}{x - y}, & x \neq y, \\
        f'\rbra{x}, & x = y.
    \end{cases}
    \]
\end{lemma}

We need the following lemma, which is a complex-valued version of \cite[Theorem 3.3]{LS01}.

\begin{lemma}\label{lem:trace-Hessian}
Let $f\in C^2((0,\infty))$, let $H$ be Hermitian, and let
$A = \diag\rbra{a_1, a_2, \dots, a_d}\succ0$.  Then
\[
\left.\frac{\mathrm{d}^2}{\mathrm{d}t^2} \tr\rbra{f\rbra{A+tH}}\right|_{t = 0} = \sum_{i=1}^d \sum_{j=1}^d \rbra{f'}^{[1]}\rbra{a_i, a_j} \abs{H_{ij}}^2.
\]
\end{lemma}

\begin{proof}
Let $g(t)\coloneqq\tr\rbra{f(A+tH)}$.
Then, using the fact that $\mathrm{D} \rbra{\tr \circ f} \rbra{A} [H] = \tr\rbra{f'\rbra{A} H}$, $g'(t)=\tr\rbra{f'(A+tH)H}$.
By \cref{lemma:dfah}, $\rbra{\mathrm{D} f'(A)[H]}_{ij}=(f')^{[1]}(a_i,a_j)H_{ij}$.
Consequently,
\begin{align*}
  g''(0)
  & = \tr\rbra*{ \mathrm{D} f'(A)[H]H } \\
  &=\sum_{i=1}^d \sum_{j=1}^d
    \rbra*{\mathrm{D} f'(A)[H]}_{ij}H_{ji}\\
  &=\sum_{i=1}^d \sum_{j=1}^d
    (f')^{[1]}(a_i,a_j)H_{ij}H_{ji}\\
  &=\sum_{i=1}^d \sum_{j=1}^d
    (f')^{[1]}(a_i,a_j)\abs{H_{ij}}^2,
\end{align*}
where the last equality uses $H_{ji}=H_{ij}^*$.
\end{proof}

Now we are ready to prove \cref{lem:one-step-bias}.

\begin{proof}[Proof of \cref{lem:one-step-bias}]
Set $\Delta=\rho-\sigma$, $Z_t=\rbra{1-t}\sigma+t\rho$, and
$g\rbra{t}=\tr\rbra{Z_t^\alpha}$.  Since the map
$A\mapsto\tr\rbra{A^\alpha}$ is concave,
\[
\tr\rbra{\rho^\alpha} = g\rbra{1}\leq g\rbra{0}+g'\rbra{0}
=\rbra{1-\alpha}\tr\rbra{\sigma^\alpha}
+\alpha\tr\rbra{\rho\sigma^{\alpha-1}}.
\]
This proves the first inequality.

For $0\leq t<1$, $Z_t\succeq\rbra{1-t}\frac{d}{n}I$. 
Let
$Z_t=\sum_{i=1}^d z_i\ketbra{v_i}{v_i}$ be the spectrum decomposition.
Applying \cref{lem:trace-Hessian} to $f\rbra{x}=x^\alpha$ in this eigenbasis gives
\begin{equation}\label{eq:part-one-second-derivative}
 g''\rbra{t}
 =\sum_{i=1}^d\sum_{j=1}^d
 \rbra{f'}^{[1]}\rbra{z_i, z_j}
 \abs*{\bra{v_i}\Delta\ket{v_j}}^2.
\end{equation}
Note that for $x > y \geq q > 0$, it holds that
\begin{equation} \label{eq:fp[1]}
-\rbra{f'}^{[1]}\rbra{x, y} 
\leq\alpha \min\cbra{x, y}^{\alpha-1} \cdot \frac{2}{x+y}.
\end{equation}
To see this, without loss of generality, we assume that $x \geq y$; denoting $x = wy$ where $w > 1$ and $\beta = 1 - \alpha$, it holds that 
\begin{align*}
-\rbra{f'}^{[1]}\rbra{x, y} 
& =-\alpha \cdot \frac{x^{\alpha-1}-y^{\alpha-1}}{x-y} \\
&=\alpha y^{-\beta-1} \cdot \frac{1-w^{-\beta}}{w-1}\\
&\leq\alpha y^{-\beta-1} \cdot \frac{2}{w+1}\\
&\leq\alpha q^{\alpha-1} \cdot \frac{2}{x+y},
\end{align*}
where the first inequality follows from $w^{-\beta}\geq w^{-1}$.
The case $x=y$ follows by continuity.

For a positive definite operator $A$, define
$L_A\rbra{H}=AH$ and $R_A\rbra{H}=HA$.  With the Hilbert--Schmidt inner
product $\ave{A,B}_{\mathrm{HS}}=\tr\rbra{A^\dag B}$, the definition of $\mathrm{D}_{\chi^2}\rbra{\rho\,\|\,\sigma}$ in
\cref{def:chi2} is equivalently $\mathrm{D}_{\chi^2}\rbra{\rho\,\|\,\sigma}=2\ave{\Delta,\rbra{L_\sigma+R_\sigma}^{-1}\rbra{\Delta}}_{\mathrm{HS}}$ (cf.\ \cite[Definition 1 and Equation (19)]{TKR+10}).
Combining \cref{eq:fp[1]} with
\cref{eq:part-one-second-derivative} gives
\[
 -g''\rbra{t}
 \leq
 \alpha\rbra*{\rbra{1-t}\frac{d}{n}}^{\alpha-1} \cdot 
 2\ave*{\Delta,
 \rbra{L_{Z_t}+R_{Z_t}}^{-1}\rbra{\Delta}}_{\mathrm{HS}}.
\]
Since $L_{Z_t}+R_{Z_t}\succeq\rbra{1-t}\rbra{L_\sigma+R_\sigma}$,
\begin{equation}\label{eq:part-one-second-derivative-bound}
 -g''\rbra{t}
 \leq
 \alpha\rbra*{\frac{d}{n}}^{\alpha-1}
 \rbra{1-t}^{\alpha-2}
 \mathrm{D}_{\chi^2}\rbra{\rho\,\|\,\sigma}.
\end{equation}

For $0<u<1$, Taylor's integral formula gives
\begin{equation} \label{eq:taylor-integral}
 g\rbra{0}+ug'\rbra{0}-g\rbra{u}
 =-\int_0^u\rbra{u-t}g''\rbra{t}\,\mathrm{d}t.
\end{equation}
We now justify the limit $u \to 1^-$ explicitly.  For $0<u<1$, extend
the integrand to the fixed interval $\sbra{0,1}$ by setting
\[
 h_u\rbra{t}:=
 \begin{cases}
  \rbra{u-t}\rbra{-g''\rbra{t}},&0\leq t\leq u,\\
  0,&u<t\leq1.
 \end{cases}
\]
Then, \cref{eq:taylor-integral} becomes
\begin{equation} \label{eq:g-eq-int-h}
 g\rbra{0}+ug'\rbra{0}-g\rbra{u}
 =\int_0^1 h_u\rbra{t}\,\mathrm{d}t.
\end{equation}
Define the limiting function by
\[
 h\rbra{t}:=
 \begin{cases}
  \rbra{1-t}\rbra{-g''\rbra{t}},&0\leq t<1,\\
  0,&t=1.
 \end{cases}
\]
For every fixed
$0\leq t<1$, once $u>t$ we have
$h_u\rbra{t}=\rbra{u-t}\rbra{-g''\rbra{t}}$, so
$h_u\rbra{t}\to h\rbra{t}$ as $u \to 1^-$.  At $t=1$, the same
convergence holds because $h_u\rbra{1}=h\rbra{1}=0$.

Moreover, concavity gives $-g''\rbra{t}\geq0$.  If $t\leq u$, then
$0\leq u-t\leq1-t$, so \cref{eq:part-one-second-derivative-bound}
implies
\begin{align*}
 0\leq h_u\rbra{t}
 &\leq\rbra{1-t}\rbra{-g''\rbra{t}}\\
 &\leq
 \alpha\rbra*{\frac{d}{n}}^{\alpha-1}
 \rbra{1-t}^{\alpha-1}\mathrm{D}_{\chi^2}\rbra{\rho\,\|\,\sigma}.
\end{align*}
The same bound is immediate when $t>u$, because then $h_u\rbra{t}=0$.
Thus the family $\cbra{h_u}_{0<u<1}$ is dominated independently of $u$
by
\[
 \Phi\rbra{t}:=
 \begin{cases}
 \alpha\rbra*{\dfrac{d}{n}}^{\alpha-1}
 \mathrm{D}_{\chi^2}\rbra{\rho\,\|\,\sigma}
 \rbra{1-t}^{\alpha-1},&0\leq t<1,\\
 0,&t=1.
 \end{cases}
\]
This function is integrable on $\sbra{0,1}$: indeed,
$\alpha-1>-1$ and
\[
 \int_0^1\Phi\rbra{t}\,\mathrm{d}t
 =\rbra*{\frac{d}{n}}^{\alpha-1}
 \mathrm{D}_{\chi^2}\rbra{\rho\,\|\,\sigma}<\infty.
\]
The dominated convergence theorem therefore gives
\begin{equation} \label{eq:limit}
 \lim_{u\to1^-}\int_0^1h_u\rbra{t}\,\mathrm{d}t
 =\int_0^1h\rbra{t}\,\mathrm{d}t
 =\int_0^1\rbra{1-t}\rbra{-g''\rbra{t}}\,\mathrm{d}t,
\end{equation}
where the value assigned to the last integrand at $t=1$ is immaterial.
On the other hand, when $u \to 1^-$, $Z_u\to\rho$ and the map
$A\mapsto\tr\rbra{A^\alpha}$ is continuous on the positive semidefinite
cone, so $g\rbra{u}\to g\rbra{1}$; hence the left-hand side of \cref{eq:taylor-integral} converges to 
\begin{equation} \label{eq:limit-g}
\lim_{u \to 1^-} \rbra*{g\rbra{0} + ug'\rbra{0} - g\rbra{u}} = g\rbra{0}+g'\rbra{0}-g\rbra{1}.
\end{equation}
Taking the limit $u \to 1^-$ in \cref{eq:g-eq-int-h} and applying
\cref{eq:part-one-second-derivative-bound,eq:limit,eq:limit-g} yields
\begin{align*}
g\rbra{0}+g'\rbra{0}-g\rbra{1}
&\leq
\alpha\rbra*{\frac{d}{n}}^{\alpha-1}
\mathrm{D}_{\chi^2}\rbra{\rho\,\|\,\sigma}
\int_0^1\rbra{1-t}^{\alpha-1}\,\mathrm{d}t\\
&=\rbra*{\frac{d}{n}}^{\alpha-1}
\mathrm{D}_{\chi^2}\rbra{\rho\,\|\,\sigma}.
\end{align*}
This proves \cref{eq:one-step-bias} by noting that $g\rbra{0}+g'\rbra{0}-g\rbra{1}=\rbra{1-\alpha}\tr\rbra{\sigma^\alpha}+\alpha\tr\rbra{\rho\sigma^{\alpha-1}}
-\tr\rbra{\rho^\alpha}$. 
\end{proof}

\subsubsection{Variance bounds}

\begin{lemma}\label{lem:correction-moments}
Let $\widehat{\rho}$ be the output matrix of $\mathsf{TomoBures\chi^2}\rbra{\rho, d, n}$ for a $d$-dimensional quantum state.
Let $\sigma = \widehat{\rho} + \frac{2d}{n}I = \sum_{i=1}^d s_i \ketbra{v_i}{v_i}$ and define a random variable $X$ by
\[
X = s_{J}^{\alpha-1}, \qquad \Pr\sbra*{ J = i } = \bra{v_i} \rho \ket{v_i}, \qquad 1 \leq i \leq d.
\]
Then, on event $E$ in \cref{lemma:bures-chi2-tomo}:
\begin{enumerate}[(i)]
 \item When $0<\alpha\leq1/2$, then there is a universal constant $C > 0$ such that
 \[
  \E\sbra*{X^2}
  \leq C d\rbra*{\frac{d}{n}}^{2\alpha-1}.
 \]
 \item When $1/2<\alpha<1$, if $n \geq d^2$, then there is a constant $C_\alpha > 0$ depending only on $\alpha$ such that
 \[
  \E\sbra*{X^2}\leq C_\alpha d^{2-2\alpha}.
 \]
 \item When $1/2<\alpha<1$, if $n^\alpha \geq d^{\alpha+1}$, then there is a constant $C_\alpha > 0$ depending only on $\alpha$ such that
 \[
  \E\sbra*{X^2}
  \leq C_\alpha d^{1/\alpha-1}
  \rbra*{\tr\rbra{\rho^\alpha}}^2.
 \]
\end{enumerate}
\end{lemma}

\begin{proof}
Set $p_i=\bra{v_i}\rho\ket{v_i}$ and $u_i=p_i-s_i= \bra{v_i}\rbra{\rho-\sigma}\ket{v_i}$.
On event $E$, \cref{lem:regularized-chi2} gives
\begin{equation}\label{eq:part-one-diagonal-control}
 \sum_{i=1}^d\frac{u_i^2}{s_i}
 \leq
 \mathrm{D}_{\chi^2}\rbra{\rho\,\|\,\sigma}
 \leq
 8\rbra{C^2+1}\frac{d^2}{n}.
\end{equation}
Moreover,
\begin{equation}\label{eq:part-one-moment-decomposition}
 \E\sbra*{X^2}=\sum_{i=1}^d p_i \rbra*{s_i^{\alpha-1}}^2
 =\sum_{i=1}^d s_i^{2\alpha-1}
  +\sum_{i=1}^d u_i s_i^{2\alpha-2}.
\end{equation}
By the Cauchy--Schwarz inequality, together with 
\cref{eq:part-one-diagonal-control}, we have
\begin{equation}\label{eq:part-one-cross-term}
 \abs*{\sum_{i=1}^d u_i s_i^{2\alpha-2}} \leq \rbra*{\sum_{i=1}^d \frac{u_i^2}{s_i}}^{1/2} \rbra*{\sum_{i=1}^d s_i^{4\alpha-3}}^{1/2}
 \leq \sqrt{8\rbra{C^2+1}}\frac{d}{\sqrt n}
 \rbra*{\sum_{i=1}^d s_i^{4\alpha-3}}^{1/2}.
\end{equation}

\textbf{Part (i).} Suppose $0<\alpha\leq1/2$.  Since $s_i\geq d/n$ and both exponents
$2\alpha-1$ and $4\alpha-3$ are nonpositive,
\[
 \sum_{i=1}^d s_i^{2\alpha-1}
 \leq d\rbra*{\frac{d}{n}}^{2\alpha-1},
 \qquad
 \sum_{i=1}^d s_i^{4\alpha-3}
 \leq d\rbra*{\frac{d}{n}}^{4\alpha-3}.
\]
Substitution into \cref{eq:part-one-moment-decomposition,eq:part-one-cross-term} gives
\[
\E\sbra*{X^2} \leq \rbra*{\sqrt{8\rbra{C^2+1}}+1} d\rbra*{\frac{d}{n}}^{2\alpha-1}.
\]

\textbf{Part (ii).} Now suppose $1/2<\alpha<1$ and $n \geq d^2$. Then,
\cref{lem:regularized-chi2} gives
\[
 \sum_{i=1}^d s_i=\tr\rbra{\sigma}\leq3,
 \qquad
 \frac{d}{n}\leq\frac1d.
\]
Because $2\alpha-1\in\rbra{0,1}$, concavity and
$\sum_{i=1}^d s_i\leq3$ give
\begin{align}
 \sum_{i=1}^d s_i^{2\alpha-1}
 \leq d^{1-(2\alpha-1)}
 \rbra*{\sum_{i=1}^d s_i}^{2\alpha-1}\leq3^{2\alpha-1}d^{2-2\alpha}.
 \label{eq:part-one-leading-large}
\end{align}
\begin{itemize}
    \item If $1/2<\alpha\leq3/4$, then $4\alpha-3\leq0$, so
$s_i\geq d/n$ implies
\[
 \sum_{i=1}^d s_i^{4\alpha-3}
 \leq d\rbra*{\frac{d}{n}}^{4\alpha-3}.
\]
Consequently, \cref{eq:part-one-cross-term} gives
\begin{align*}
 \abs*{\sum_{i=1}^d u_i s_i^{2\alpha-2}}
 &\leq\sqrt{8\rbra{C^2+1}}\frac{d}{\sqrt n}
 \rbra*{d\rbra*{\frac{d}{n}}^{4\alpha-3}}^{1/2}\\
 &=\sqrt{8\rbra{C^2+1}}d
 \rbra*{\frac{d}{n}}^{2\alpha-1}\\
 &\leq\sqrt{8\rbra{C^2+1}}d^{2-2\alpha},
\end{align*}
where the last inequality uses $d/n\leq1/d$.

    \item If $3/4\leq\alpha<1$, then $4\alpha-3\in\sbra{0,1}$.  Concavity gives
\begin{align*}
 \sum_{i=1}^d s_i^{4\alpha-3}
 \leq d^{1-(4\alpha-3)}
 \rbra*{\sum_{i=1}^d s_i}^{4\alpha-3}\leq3^{4\alpha-3}d^{4-4\alpha}.
\end{align*}
It follows from $d/\sqrt n\leq1$ and
\cref{eq:part-one-cross-term} that
\[
 \abs*{\sum_{i=1}^d u_i s_i^{2\alpha-2}}
 \leq\sqrt{8\rbra{C^2+1}}3^{(4\alpha-3)/2}
 d^{2-2\alpha}.
\]
\end{itemize}
Combining the above two cases with
\cref{eq:part-one-leading-large,eq:part-one-moment-decomposition} gives
\begin{align*}
 \E\sbra*{X^2}
 &\leq\rbra*{3^{2\alpha-1}
 +\sqrt{8\rbra{C^2+1}}
 \max\cbra*{1,3^{(4\alpha-3)/2}}}
 d^{2-2\alpha}.
\end{align*}

\textbf{Part (iii).} Now suppose $1/2<\alpha<1$ and $n^\alpha \geq d^{\alpha+1}$. Let
\[
 S=\sum_{i=1}^d s_i^\alpha,
 \qquad
 G=\rbra{1-\alpha}S
 +\alpha\tr\rbra{\rho\sigma^{\alpha-1}}.
\]
Both terms in $G$ are nonnegative.  Hence
$S\leq G/\rbra{1-\alpha}$; moreover, \cref{lem:one-step-bias} gives $G\geq \tr\rbra{\rho^{\alpha}} \geq1$.
Since $0<2\alpha-1<\alpha$, H\"older's inequality gives
\begin{align}
 \sum_{i=1}^d s_i^{2\alpha-1}
 \leq d^{1-(2\alpha-1)/\alpha}
 \rbra*{\sum_{i=1}^d s_i^\alpha}^{(2\alpha-1)/\alpha}.
 \label{eq:part-one-relative-leading}
\end{align}

\begin{itemize}
    \item If $1/2<\alpha<3/4$, then
$s_i^{4\alpha-3}\leq\rbra{d/n}^{2\alpha-2}s_i^{2\alpha-1}$.
By \cref{eq:part-one-relative-leading,eq:part-one-cross-term}, 
we have 
\begin{align*}
 \abs*{\sum_{i=1}^d u_i s_i^{2\alpha-2}}
 &\leq\sqrt{8\rbra{C^2+1}}\,
 d^{1/(2\alpha)}
 \rbra*{\frac{d}{n}}^{\alpha-1/2}
 S^{(2\alpha-1)/(2\alpha)}\\
 &\leq\sqrt{8\rbra{C^2+1}}\,
 d^{1/\alpha-1}S^{(2\alpha-1)/(2\alpha)}\\
 &\leq
 \sqrt{8\rbra{C^2+1}}\,
 \rbra{1-\alpha}^{-(2\alpha-1)/(2\alpha)}
 d^{1/\alpha-1}G^2.
\end{align*}
The last inequality again follows from
$S\leq G/\rbra{1-\alpha}$ and $G\geq1$.
\item If $3/4\leq\alpha<1$, put $r=4\alpha-3$.  For $r>0$, H\"older's
inequality gives
$\sum_i s_i^r\leq d^{1-r/\alpha}S^{r/\alpha}$; for $r=0$, the same formula
is the identity $\sum_i s_i^0=d$.  Substituting this bound into
\cref{eq:part-one-cross-term} gives
\begin{align*}
 \abs*{\sum_{i=1}^d u_i s_i^{2\alpha-2}}
 &\leq\sqrt{8\rbra{C^2+1}}\,
 d^{1-r/(2\alpha)}\rbra*{\frac{d}{n}}^{1/2}
 S^{r/(2\alpha)}\\
 &\leq\sqrt{8\rbra{C^2+1}}\,
 d^{1/\alpha-1}S^{r/(2\alpha)}\\
 &\leq
 \sqrt{8\rbra{C^2+1}}\,
 \rbra{1-\alpha}^{-r/(2\alpha)}
 d^{1/\alpha-1}G^2.
\end{align*}
Here the second inequality uses $d/n\leq d^{-1/\alpha}$ and
$r+1=4\alpha-2$, while the last inequality uses
$S\leq G/\rbra{1-\alpha}$ and $G^{r/(2\alpha)}\leq G^2$.
\end{itemize}

Combining the above two cases with
\cref{eq:part-one-leading-large,eq:part-one-moment-decomposition} gives
\begin{align*}
 \E\sbra*{X^2}
 &\leq\rbra*{3^{2\alpha-1}
 +\sqrt{8\rbra{C^2+1}}
 \max\cbra*{1,3^{(4\alpha-3)/2}}}
 d^{2-2\alpha}.
\end{align*}
Since $4\alpha-3\leq2\alpha-1$, the above two cases gives
\[
 \abs*{\sum_{i=1}^d u_i s_i^{2\alpha-2}} \leq \sqrt{8\rbra{C^2+1}}\,
 \rbra{1-\alpha}^{-(2\alpha-1)/(2\alpha)}
 d^{1/\alpha-1}G^2.
\]
Together with
\cref{eq:part-one-relative-leading,eq:part-one-moment-decomposition},
this gives
\begin{align}
 \E\sbra*{X^2}
 \leq
 \rbra*{
 \rbra{1-\alpha}^{-(2\alpha-1)/\alpha}
 +\sqrt{8\rbra{C^2+1}}\,
 \rbra{1-\alpha}^{-(2\alpha-1)/(2\alpha)} } d^{1/\alpha-1}G^2. \label{eq:EX2-iii}
\end{align}
Moreover, \cref{lem:regularized-chi2,lem:one-step-bias} and
$n^\alpha\geq d^{\alpha+1}$ give
\begin{align*}
 0
 \leq G-\tr\rbra{\rho^\alpha}\leq
 \rbra*{\frac{d}{n}}^{\alpha-1}
 \mathrm{D}_{\chi^2}\rbra{\rho\,\|\,\sigma}\leq8\rbra{C^2+1}\frac{d^{\alpha+1}}{n^\alpha}
 \leq8\rbra{C^2+1}.
\end{align*}
Since $\mathrm{F}_\alpha\rbra{\rho}\geq1$, it follows that $G\leq\rbra{8\rbra{C^2+1}+1}
 \tr\rbra{\rho^\alpha}$.
Substituting this into \cref{eq:EX2-iii} gives $\E\sbra{X^2}
  \leq C_\alpha d^{1/\alpha-1}
  \rbra{\tr\rbra{\rho^\alpha}}^2$ with
\begin{align*}
 C_\alpha
 ={}&\rbra*{8\rbra{C^2+1}+1}^2
 \rbra*{
 \rbra{1-\alpha}^{-(2\alpha-1)/\alpha}
 +\sqrt{8\rbra{C^2+1}}\,
 \rbra{1-\alpha}^{-(2\alpha-1)/(2\alpha)}
 }.
\end{align*}
\end{proof}

\subsubsection{Proof of Lemma \ref{thm:common-estimator}} \label{sec:proof-common-estimator}

Using the above technical lemmas, we prove \cref{thm:common-estimator} as follows. 

\begin{proof}[Proof of \cref{thm:common-estimator}]
Let $\widehat{\rho}$ be the tomography output of $\mathsf{TomoBures\chi^2}\rbra{\rho, d, n}$ and write
$\sigma=\widehat{\rho}+\frac{2d}{n}I=\sum_{i=1}^d
s_i\ketbra{v_i}{v_i}$, as in \cref{alg:powertrace}.  
Note that $X_j$ for $1 \leq j \leq m$ are independent and identically distributed, with
\[
 \Pr\sbra*{J_j=i\mid\sigma}=\bra{v_i}\rho\ket{v_i},
 \qquad X_j=s_{J_j}^{\alpha-1}.
\]
Consequently,
\begin{equation}\label{eq:part-one-conditional-mean}
 \E\sbra*{\widetilde{\mathrm{F}}_\alpha}
 =\rbra{1-\alpha}\tr\rbra{\sigma^\alpha}
 +\alpha\tr\rbra{\rho\sigma^{\alpha-1}}.
\end{equation}
Let $E$ be the event in \cref{lemma:bures-chi2-tomo} with
$\Pr\sbra{E}\geq0.99$. 
On event $E$,
\cref{lem:regularized-chi2,lem:one-step-bias} and
\cref{eq:part-one-conditional-mean} give
\begin{align}
 0\leq \E\sbra*{\widetilde{\mathrm{F}}_\alpha}-F
 \leq
 \rbra*{\frac{d}{n}}^{\alpha-1}
 \mathrm{D}_{\chi^2}\rbra{\rho\,\|\,\sigma}
 =O_\alpha\!\left(\frac{d^{\alpha+1}}{n^\alpha}\right)
 =O_\alpha\rbra{\varepsilon}.
 \label{eq:part-one-bias}
\end{align}
On the other hand, 
\begin{equation*}
 \Var\sbra*{\widetilde{\mathrm{F}}_\alpha}
 =\frac{\alpha^2}{m}\Var\sbra*{X_1}
 \leq\frac{\alpha^2}{m}\E\sbra*{X_1^2}.
\end{equation*}
By Chebyshev's inequality,
\begin{equation} \label{eq:prob-cheby}
\Pr\sbra*{ \abs*{\widetilde{\mathrm{F}}_\alpha-\E\sbra*{\widetilde{\mathrm{F}}_\alpha}} \leq t } \geq 1 - \frac{\Var\sbra{\widetilde{\mathrm{F}}_\alpha}}{t^2} \geq 1 - \frac{\alpha^2}{m t^2} \E\sbra*{X_1^2}.
\end{equation}

\textbf{Part (i).} Suppose first that $0<\alpha\leq1/2$.  Part~(i) of
\cref{lem:correction-moments} gives that on $E$, $\E\sbra{X_1^2} \leq O\rbra{d^{2\alpha}n^{1-2\alpha}}$.
By \cref{eq:prob-cheby}, with appropriate choice of $C_{\alpha}^{(0)}$ and $C_{\alpha}^{(1)}$, we have $\Pr\sbra{\abs{\widetilde{\mathrm{F}}_\alpha - \E\sbra{\widetilde{\mathrm{F}}_\alpha}} \leq \varepsilon/2} \geq 0.99$. 
Together with \cref{eq:part-one-bias}, this yields $\Pr\sbra{\abs{\widetilde{\mathrm{F}}_\alpha - F} \leq \varepsilon} \geq 0.98$. 

\textbf{Part (ii).} Suppose that $1/2<\alpha<1$.  The choice of $n$ satisfies $n\geq d^{1+1/\alpha}\geq d^2$.
Hence part~(ii) of \cref{lem:correction-moments} gives that on $E$, $\E\sbra{X_1^2}\leq O_\alpha \rbra{d^{2-2\alpha}}$.
By \cref{eq:prob-cheby}, with $t=\varepsilon/2$ and
$C_\alpha^{(1)}$ sufficiently large, we have on $E$, $\Pr\sbra{
 \abs{\widetilde{\mathrm{F}}_\alpha-\E\sbra{\widetilde{\mathrm{F}}_\alpha}}
 \leq{\varepsilon}/{2}}
 \geq0.99$.
Together with \cref{eq:part-one-bias}, this gives on $E$, $\Pr\sbra{
 \abs{\widetilde{\mathrm{F}}_\alpha-F}
 \leq\varepsilon}
 \geq0.98$.

\textbf{Part (iii).} Suppose that $1/2<\alpha<1$.  The choice of
$n$ satisfies $n^\alpha\geq d^{\alpha+1}$.  Hence part~(iii) of
\cref{lem:correction-moments} and \cref{eq:global-physical-range} give
that on $E$, $\E\sbra*{X_1^2}
 \leq C_\alpha d^{1/\alpha-1}F^2$.
By \cref{eq:prob-cheby}, with $t=\varepsilon F/2$ and
$C_\alpha^{(1)}$ sufficiently large, we have on $E$, $\Pr\sbra{
 \abs{\widetilde{\mathrm{F}}_\alpha-\E\sbra{\widetilde{\mathrm{F}}_\alpha}}
 \leq{\varepsilon F}/{2}}
 \geq0.99$.
Since \cref{eq:part-one-bias} gives
$0\leq \E\sbra{\widetilde{\mathrm{F}}_\alpha}-F\leq\varepsilon/2\leq\varepsilon F/2$, we obtain on $E$, $\Pr\sbra{
 \abs{\widetilde{\mathrm{F}}_\alpha-F}
 \leq\varepsilon F}
 \geq0.98$.
\end{proof}

\section{R\'enyi entropy with \texorpdfstring{$\alpha>1$}{α > 1}}

In this section, we prove the following result.

\begin{theorem}[R\'enyi entropy estimation for $\alpha > 1$] \label{gt:thm:main}
    For non-integer $\alpha > 1$, given sample access to an unknown $d$-dimensional quantum state $\rho$ and the additive error $\varepsilon$, we can estimate the $\alpha$-R\'enyi entropy $\mathrm{S}_\alpha^{\mathrm{R}}\rbra{\rho}$ using 
    \[
    O\rbra*{\frac{d^2}{\varepsilon^{1/\alpha}}+\frac{d^{1-1/\alpha}}{\varepsilon^2}}
    \]
    samples of $\rho$. 
\end{theorem}

\subsection{The estimator}
    First, our estimator requires the following subroutines.

    \begin{lemma}[Random purification, {\cite[Lemma 2.11]{tang2025conjugate}}]\label{lemma-8170436}
        For every integer $s\ge1$, there is a channel $\Lambda_{\mathrm{pur}}^{(s)}$ such that for every $d$-dimensional quantum state $\rho$,
        \[\Lambda_{\mathrm{pur}}^{(s)}(\rho^{\otimes s})=\E_{\ket{\psi_\rho}}\sbra*{\ketbra{\psi_\rho}{\psi_\rho}^{\otimes s}},\]
where $\ket{\psi_\rho}\in\mathbb{C}^d\otimes\mathbb{C}^d$ is sampled uniformly from the space of purifications of $\rho$, meaning that $\tr_\textup{anc}(\ketbra{\psi_\rho}{\psi_\rho})=\rho$.
    \end{lemma}

    \begin{definition}[Hayashi's POVM \cite{Hay98}]\label{def-8170435}
        For a pure state $\ket{\psi}\in\C^{d^2}$, Hayashi's covariant POVM on
        $\operatorname{Sym}^s(\C^{d^2})$ is
        \begin{equation}\label{gt:eq:hayashi-povm}
        \dim\bigl(\operatorname{Sym}^s(\C^{d^2})\bigr)
         \ket{v}\!\bra{v}^{\otimes s}\,\mathrm{d}v = \binom{d^2+s-1}{s} \ketbra{v}{v}^{\otimes s} \mathrm{d}v,
        \end{equation}
        where $\mathrm{d}v$ is Haar probability measure.
        Note that on input $\ket{\psi}^{\otimes s}$, the outcome density is proportional to $|\braket{v}{\psi}|^{2s}$.
    \end{definition}
    
Now we are ready to define our estimator. 
Fix an order $\alpha>1$ and an accuracy $0<\varepsilon\le1$.  Set
\begin{equation}\label{gt:eq:parameters}
 \theta:=1-e^{-(\alpha-1)\varepsilon},
 \qquad
 k:=\lceil\alpha\rceil-1.
\end{equation}
Fix an $\alpha$-dependent constant $K_\alpha\ge1$ (defined later in \cref{eq-8170410}) and define
\begin{equation}\label{gt:eq:batch-sizes}
 m:=\ceil*{ K_\alpha\rbra*{
 \frac{d^2}{\theta^{1/\alpha}}
 +\frac{d^{1-1/\alpha}}{\theta^{2}}
 } },
 \qquad
 s_\ell:=2^\ell m
 \quad(0\le\ell\le k).
\end{equation}
Because $\theta\le1$, every $s_\ell$ is at least $d^2$.
For integers $j\ge0$ and $s\ge1$, put
\begin{equation}\label{gt:eq:mu-def}
 \mu_0(s):=1,
 \qquad
 \mu_j(s):=\prod_{r=0}^{j-1}
 \frac{d^2-1+r}{d^2+s+r}
 \quad(j\ge1).
\end{equation}
Let $a_0,\ldots,a_k$ be the unique solution of
\begin{equation}\label{gt:eq:rich-system}
 \sum_{\ell=0}^k a_\ell=1,
 \qquad
 \sum_{\ell=0}^k a_\ell\mu_j(s_\ell)=0
 \quad(1\le j\le k).
\end{equation}
The existence, uniqueness, and uniform boundedness of these coefficients are proved in
Lemma~\ref{gt:lem:richardson}.

\begin{algorithm}[htbp]
\caption{$\mathsf{RenyiEntropy}_{>1}\rbra{\rho,d,\alpha,\varepsilon}$}
\label{gt:alg:main}
\begin{algorithmic}[1]
\Require Independent samples of an unknown $d$-dimensional quantum state $\rho$; $\alpha > 1$; $0<\varepsilon\le1$.
\State Compute $\theta,k,m,s_0,\ldots,s_k$ from
\cref{gt:eq:parameters,gt:eq:batch-sizes}.
\State Compute $\mu_j(s_\ell)$ from \cref{gt:eq:mu-def} and solve
\cref{gt:eq:rich-system} for $a_0,\ldots,a_k$.
\For{$\ell=0,\ldots,k$}
  \State Reserve a fresh batch of $s_\ell$ samples of $\rho$.
  \State Apply the random-purification channel $\Lambda_{\mathrm{pur}}^{(s_\ell)}$ to this
  batch. \Comment{see \cref{lemma-8170436}}
  \State Apply the covariant pure-state POVM to the channel output and denote its outcome by
  $\ket{v_\ell}\in\C^d\otimes\C^d$. \Comment{see \cref{def-8170435}}
  \State Set $\widehat\rho_\ell\gets\tr_{\mathrm{anc}}\ket{v_\ell}\!\bra{v_\ell}$ and
  $Y_\ell\gets\tr(\widehat\rho_\ell^\alpha)$.
\EndFor
\State Set $\widetilde{\mathrm{F}}_\alpha\gets\sum_{\ell=0}^k a_\ell Y_\ell$.
\State Set
$\widehat{\mathrm{F}}_\alpha\gets
\operatorname{Clip}_{[d^{1-\alpha},1]}(\widetilde{\mathrm{F}}_\alpha)$.
\State \Return
$\displaystyle
\widehat{\mathrm{S}}_\alpha^{\mathrm{R}}
\gets\frac{\log\rbra{\widehat{\mathrm{F}}_\alpha}}{1-\alpha}$.
\end{algorithmic}
\end{algorithm}

\subsection{Main proof}

For this section only, abbreviate $F:=\mathrm{F}_\alpha(\rho)$. 
To prove \cref{gt:thm:main}, we need the following two lemmas.

\begin{lemma}[One-batch estimates]\label{gt:prop:one-batch}
Let $s\ge d^2$, and let $Y_s$ be the power estimate obtained from one batch of size $s$ in
Algorithm~\ref{gt:alg:main}.  There are coefficients $c_1(\rho),\ldots,c_k(\rho)$ independent
of $s$, such that
\begin{equation}\label{gt:eq:one-batch-mean}
 \E\sbra*{Y_s}
 =F+\sum_{j=1}^k c_j(\rho)\mu_j(s)+R_s,
 \qquad
 |c_j(\rho)|\le O_\alpha(F),
 \qquad
 |R_s|\le O_{\alpha}\!\left(F\cdot \left(\frac{d^2}{s}\right)^\alpha\right).
\end{equation}
Moreover,
\begin{equation}\label{gt:eq:one-batch-var}
 \Var\sbra*{Y_s}
 \le O_\alpha\!\left(F^2\cdot \left[
 \frac{d^{1-1/\alpha}}{s}
 +\left(\frac{d^2}{s}\right)^{2\alpha}
 \right]\right).
\end{equation}
\end{lemma}

The proof of \cref{gt:prop:one-batch} is given in \cref{gt:sec:one-batch-proof}.

\begin{lemma}[Bias and variance]\label{gt:prop:combined}
The raw estimator produced by Algorithm~\ref{gt:alg:main} satisfies: there exists constants $C_{\alpha}^{(1)}$ and $C_{\alpha}^{(2)}$ depending only on $\alpha$, such that
\begin{align}
 \left|\E\sbra*{\widetilde{\mathrm{F}}_\alpha}-F\right|
 &\le C_\alpha^{(1)} F\left(\frac{d^2}{m}\right)^\alpha,
 \label{gt:eq:combined-bias}\\
 \Var\sbra{\widetilde{\mathrm{F}}_\alpha}
 &\le C_\alpha^{(2)} F^2\left[
 \frac{d^{1-1/\alpha}}{m}
 +\left(\frac{d^2}{m}\right)^{2\alpha}
 \right].
 \label{gt:eq:combined-var}
\end{align}
\end{lemma}

\begin{proof}
By \cref{gt:prop:one-batch},
\[
 \E\sbra*{Y_\ell}
 =F+\sum_{j=1}^k c_j(\rho)\mu_j(s_\ell)+R_{s_\ell}.
\]
The coefficients $c_j(\rho)$ are common to all batches.  Hence
\cref{gt:eq:rich-system} cancels the displayed polynomial terms and gives
\[
 \E\sbra*{\widetilde{\mathrm{F}}_\alpha}-F
 =\sum_{\ell=0}^k a_\ell R_{s_\ell}.
\]
Using $s_\ell\ge m$ with \cref{gt:eq:one-batch-mean}, and the coefficient bound in Lemma~\ref{gt:lem:richardson}, we can prove
\cref{gt:eq:combined-bias}.

The batches use disjoint input samples and separate channels and measurements, so
$Y_0,\ldots,Y_k$ are independent.  Therefore
\[
 \Var\sbra*{\widetilde{\mathrm{F}}_\alpha}
 =\sum_{\ell=0}^k a_\ell^2\Var\sbra*{Y_\ell}.
\]
Now apply \cref{gt:eq:one-batch-var}, $s_\ell\ge m$, and
Lemma~\ref{gt:lem:richardson}.
\end{proof}

Then, we are able to prove \cref{gt:thm:main}.
\begin{proof}[Proof of \cref{gt:thm:main}]
Let $C_\alpha$ be a constant dominating the constants in \cref{gt:prop:combined}.  Choose the fixed tuning
constant $K_\alpha$ large enough that
\begin{equation}\label{eq-8170410}
 C_\alpha K_\alpha^{-\alpha}\le\frac14,
 \qquad
 C_\alpha\bigl(K_\alpha^{-1}+K_\alpha^{-2\alpha}\bigr)\le\frac1{100}.
\end{equation}
This choice depends only on $\alpha$.  From \cref{gt:eq:batch-sizes},
\[
 \left(\frac{d^2}{m}\right)^\alpha
 \le K_\alpha^{-\alpha}\theta,
 \qquad
 \frac{d^{1-1/\alpha}}{m}
 \le K_\alpha^{-1}\theta^2,
 \qquad
 \left(\frac{d^2}{m}\right)^{2\alpha}
 \le K_\alpha^{-2\alpha}\theta^2.
\]
Consequently, \cref{gt:prop:combined} yields
\begin{equation}\label{gt:eq:final-moments}
 \left|\E\sbra*{\widetilde{\mathrm{F}}_\alpha}-F\right|
 \le\frac14\theta F,
 \qquad
 \Var\sbra{\widetilde{\mathrm{F}}_\alpha}
 \le\frac1{100}\theta^2F^2.
\end{equation}
Chebyshev's inequality gives
$\mathbb{P}\sbra{
 \abs{\widetilde{\mathrm{F}}_\alpha-\E\sbra{\widetilde{\mathrm{F}}_\alpha}} >\frac12\theta F }\le\frac{1}{25}$.
On the complementary event,
$\abs{\widetilde{\mathrm{F}}_\alpha-F}
 \le\frac34\theta F
 \le\theta F$.
Since $F\in[d^{1-\alpha},1]$, the clip operation into this interval does not increase the
error.  Thus
$\abs{\widehat{\mathrm{F}}_\alpha-F}\le\theta F$.
The projected estimate is positive, and
\cref{lem:global-log-conversion}, with
$\gamma=(\alpha-1)\varepsilon$ and
$\theta=1-e^{-\gamma}$, implies
$\abs{\log\widehat{\mathrm{F}}_\alpha-\log F}
 \le(\alpha-1)\varepsilon$.
Dividing by $\alpha-1$ proves the asserted entropy error.  The success probability is at least
$24/25$, and hence at least $2/3$.

Finally, note that $\sum_{\ell=0}^k s_\ell=(2^{k+1}-1)m = O\rbra{m}$, and $k$ depends only on $\alpha$. Moreover, $\theta=\Theta(\varepsilon)$. Therefore, the final sample complexity is 
\[
O(m)=O\rbra*{
 \frac{d^2}{\varepsilon^{1/\alpha}}
 +\frac{d^{1-1/\alpha}}{\varepsilon^{2}}
 }.
\]
\end{proof}

\subsection{Proof of the one-batch estimates}\label{gt:sec:one-batch-proof}

Now, we give the proof of \cref{gt:prop:one-batch}.
Fix a purification $\ket{\psi}\in\mathbb{C}^{d}\otimes \mathbb{C}^d$ of $\rho$, we identify it as $\kett{M}$ for a coefficient matrix
$M\in\C^{d\times d}$.  Then
\begin{equation*}
 \|M\|_2=\|M\|_{\mathrm F}=1,
 \qquad
 \rho=MM^\dag,
 \qquad
 F=\mathrm{F}_\alpha(\rho)=\|M\|_{2\alpha}^{2\alpha}.
\end{equation*}
By Lemma~\ref{gt:lem:beta-sphere}, the coefficient matrix of the outcome of the covariant pure-state POVM (see line 6 of \cref{gt:alg:main}) is $\widehat M_T=\sqrt{1-T}\,M+\sqrt T\,G$,
where $T\sim\operatorname{Beta}(d^2-1,s+1)$ and $G$ is uniform on the unit sphere of
$M^\perp$, independently of $T$.  The one-batch statistic is
$Y_s=\|\widehat M_T\|_{2\alpha}^{2\alpha}$.

\subsubsection{Bias}

For $0\le t\le1/2$, set
$u:=\sqrt{t/(1-t)}\le1$.  By homogeneity,
\begin{equation}\label{gt:eq:scaled-Phi}
 \E_G\sbra*{Y_s\mid T=t}
 =\E_G\sbra*{\|\sqrt{1-t}M+\sqrt{t}G\|_{2\alpha}^{2\alpha}}=(1-t)^\alpha\E_G\sbra*{\|M+uG\|_{2\alpha}^{2\alpha}}.
\end{equation}
Since $\|M\|_p\le\|M\|_{\mathrm F}=1$ and
$\|G\|_p\le\|G\|_{\mathrm F}=1$, Lemma~\ref{gt:lem:schatten-taylor} applies.  The law of $G$
is invariant under $G\mapsto-G$, so every odd multilinear term has zero expectation.  The
largest even integer not exceeding $\lceil2\alpha\rceil-1$ is $2k$ (see \cref{gt:eq:parameters}).  Thus, with
$b_j:=\E_G\sbra{L_{2j,M}(G,\ldots,G)}$,
\begin{equation}\label{gt:eq:u-expansion}
 \E_G\sbra*{\|M+uG\|_{2\alpha}^{2\alpha}}
 =F+\sum_{j=1}^k b_j u^{2j}+\mathcal E(u).
\end{equation}
The multilinear bound in Lemma~\ref{gt:lem:schatten-taylor} and
Lemma~\ref{gt:lem:isotropic} give
\begin{align*}
 |b_j|
 \le O_\alpha\!\left(\|M\|_{2\alpha}^{2\alpha-2j}\E\sbra*{\|G\|_{2\alpha}^{2j}}\right)
 \le O_\alpha\!\left( F^{1-j/\alpha}d^{j(1/\alpha-1)}\right)
 \le O_\alpha(F),
\end{align*}
where the last step uses
$d^{1/\alpha-1}\le F^{1/\alpha}$.  Likewise,
\begin{equation}\label{gt:eq:u-rem}
 |\mathcal E(u)|
 \le O_\alpha\!\left(u^{2\alpha}\E\sbra*{\|G\|_{2\alpha}^{2\alpha}}\right)
 \le O_\alpha\!\left(Fu^{2\alpha}\right).
\end{equation}
Substituting \cref{gt:eq:u-expansion,gt:eq:u-rem} into
\cref{gt:eq:scaled-Phi} yields, for $0\le t\le1/2$,
\begin{equation}\label{gt:eq:t-functions}
 \E_G\sbra*{Y_s\mid T=t}
 =F(1-t)^\alpha
 +\sum_{j=1}^k b_jt^j(1-t)^{\alpha-j}
 +O_\alpha(Ft^\alpha).
\end{equation}
For each $0\le j\le k$, apply Lemma~\ref{gt:lem:endpoint-scalar} with
$\beta=\alpha-j$.  Since
$\lceil\alpha-j\rceil-1=k-j$, multiplying the resulting polynomial by $t^j$ produces a
polynomial of degree at most $k$ and an error of order $t^\alpha$.  Hence there are coefficients
$c_1(\rho),\ldots,c_k(\rho)$ with $|c_j(\rho)|\le O_\alpha(F)$ such that
\begin{equation}\label{gt:eq:conditional-mean-final}
 \E_G\sbra*{Y_s\mid T=t}
 =F+\sum_{j=1}^k c_j(\rho)t^j+r(t),
 \qquad
 |r(t)|\le O_\alpha(Ft^\alpha)
\end{equation}
for $0\le t\le1/2$.

The same bound extends to $1/2\le t\le1$.  Indeed, the Schatten triangle inequality and
Lemma~\ref{gt:lem:isotropic} imply
\[
 \E_G\sbra*{Y_s\mid T=t}= \E_G\sbra*{\|\sqrt{1-t}M+\sqrt tG\|_{2\alpha}^{2\alpha}}
 \le O_\alpha\bigl(F+\E\sbra*{\|G\|_{2\alpha}^{2\alpha}}\bigr)
 \le O_\alpha (F).
\]
The polynomial in \cref{gt:eq:conditional-mean-final} is also bounded by $O_\alpha(F)$ on
$[0,1]$.
For $1/2\leq t \leq 1$, let
\[r(t)\coloneqq \E_G\sbra*{Y_s\mid T=t}- \left(F+\sum_{j=1}^k c_j(\rho)t^j\right).\]
Thus $|r(t)|\leq O_\alpha(F) \leq O_\alpha(F t^\alpha)$ since $t^\alpha\ge2^{-\alpha}$ on $[1/2,1]$. Therefore \cref{gt:eq:conditional-mean-final} holds for all $0\leq t\leq 1$.

Averaging \cref{gt:eq:conditional-mean-final} over $T$ and using
Lemma~\ref{gt:lem:beta-sphere} proves \cref{gt:eq:one-batch-mean}.  Finally, replacing $M$ by $MU$
for a unitary $U$ on the purifying register sends $G$ to $GU$, preserves the Frobenius
hyperplane and all Schatten norms, and therefore leaves the conditional law and the
coefficients unchanged.  Thus the coefficients depend only on $\rho$, and the same expansion
holds after the random-purification channel.

\subsubsection{Variance}

For fixed $t$, write
\[
 h_t(G):=\|\sqrt{1-t}M+\sqrt tG\|_{2\alpha}^{2\alpha},
 \qquad
 \phi(t):=\E_G\sbra*{h_t(G)}.
\]
The law of total variance gives
\begin{equation}\label{gt:eq:total-var}
 \Var\sbra{Y_s}
 =\E_T\sbra*{\Var_G\sbra{h_T(G)\mid T}}
 +\Var_T\sbra{\phi(T)}.
\end{equation}

For the first term, apply
\cref{gt:lem:spherical-poincare} with $H=M^\perp$, viewed as a real
Hilbert space with inner product $\langle A,B\rangle_{\mathbb R} \coloneqq \Real\tr(A^\dag B)$.
The complex dimension of $M^\perp$ is $d^2-1$, so its real dimension is
$m=2d^2-2$.
Consequently, its unit sphere has real dimension $m-1=2d^2-3$.
For fixed $t$, consider the ambient extension
\[
 \widetilde h_t(Z)
 :=
 \|\sqrt{1-t}M+\sqrt tZ\|_{2\alpha}^{2\alpha},
 \qquad Z\in M^\perp,
\]
so that $h_t$ is the restriction of $\widetilde h_t$ to the unit
sphere of $M^\perp$.
By \cref{gt:eq:spherical-poincare-ambient},
\[
 \Var_G\sbra*{h_t(G)}
 \leq
 \frac{1}{2d^2-3}
 \E_G\sbra*{
 \left\|
 \nabla_{M^\perp}\widetilde h_t(G)
 \right\|_{\mathrm F}^2}.
\]
The gradient in $M^\perp$ is the orthogonal projection of the full
matrix gradient. 
We set \[\widehat M_t:=\sqrt{1-t}M+\sqrt tG.\]
Thus, by the chain rule,
\[
 \left\|
 \nabla_{M^\perp}\widetilde h_t(G)
 \right\|_{\mathrm F}
 \leq
 \sqrt t\,
 \left\|
 \nabla \|X\|_{2\alpha}^{2\alpha}
 \big|_{X=\widehat M_t}
 \right\|_{\mathrm F}.
\]
Using the Frobenius-gradient identity
\[
 \left\|
 \nabla_X\|X\|_{2\alpha}^{2\alpha}
 \right\|_{\mathrm F}^2
 =
 (2\alpha)^2
 \tr\!\left[(XX^\dag)^{2\alpha-1}\right],
\]
we obtain
\begin{align}
 \Var_G\sbra{h_t(G)}
 &\leq
 \frac{t}{2d^2-3}
 \E_G\sbra*{
 \left\|
 \nabla_X\|X\|_{2\alpha}^{2\alpha}
 \big|_{X=\widehat M_t}
 \right\|_{\mathrm F}^2}
 \nonumber\\
 &=
 \frac{(2\alpha)^2t}{2d^2-3}
 \E_G\sbra*{
 \tr\!\left[
 (\widehat M_t\widehat M_t^\dag)^{2\alpha-1}
 \right]}
 \nonumber\\
 &\leq
 O_\alpha\!\left(
 \frac{t}{d^2}
 \E_G\sbra*{
 \tr\!\left[
 (\widehat M_t\widehat M_t^\dag)^{2\alpha-1}
 \right]}
 \right).\nonumber
\end{align}

The Schatten triangle inequality, followed by Lemma~\ref{gt:lem:isotropic} with
$q=4\alpha-2$, yields
\begin{align}
 \E_G\sbra*{\tr\!\left[(\widehat M_t\widehat M_t^\dag)^{2\alpha-1}\right]}
 & = \E_G\sbra*{\|\sqrt{1-t}M+\sqrt{t}G\|_{4\alpha-2}^{4\alpha-2}} \nonumber \\ 
 & \leq O_\alpha\!\left((1-t)^{2\alpha-1}\E_G\sbra*{\|M\|_{4\alpha-2}^{4\alpha-2}} + t^{2\alpha-1}\E_{G}\sbra*{\|G\|_{4\alpha-2}^{4\alpha-2}}\right)\nonumber \\
 & = O_\alpha\!\left(
 \tr(\rho^{2\alpha-1})+t^{2\alpha-1}d^{2-2\alpha} \right).  \nonumber
\end{align}
Average over $T$ and use Lemma~\ref{gt:lem:beta-sphere}.  Since
$\E\sbra{T}/d^2\le O(1/s)$, this gives
\begin{align}
 \E_T\sbra*{\Var_G\sbra{h_T(G)\mid T}}
 &\le O_\alpha\!\left(
 \frac{\tr(\rho^{2\alpha-1})}{s}
 +\frac{d^{2-2\alpha}}{d^2}
  \left(\frac{d^2}{s}\right)^{2\alpha}
 \right) \nonumber \\
 &\le O_\alpha\!\left(F^2\left[
 \frac{d^{1-1/\alpha}}s
 +\left(\frac{d^2}s\right)^{2\alpha}
 \right]\right),\label{gt:eq:angular-final}
\end{align}
where the last step uses Lemma~\ref{gt:lem:relative-gradient} and
$d^{2-2\alpha}\le F^2$.  

For the second term, \cref{gt:eq:conditional-mean-final} gives
\[
 \phi(T)=F+\sum_{j=1}^k c_j(\rho)T^j+r(T),
 \qquad 
 |c_j(\rho)|\leq O_\alpha(F),\qquad
 |r(T)|\le O_\alpha(FT^\alpha).
\]
Because $k$ is fixed, Lemma~\ref{gt:lem:beta-variance} and
$\Var\sbra{r(T)}\le\E\sbra{r(T)^2}$ imply
\begin{align}
 \Var_T\sbra{\phi(T)}& = \Var_T\sbra{\sum_{j=1}^k c_j(\rho)T^j+r(T)} \nonumber\\
 &\le O_\alpha\!\left(F^2\left[
 \sum_{j=1}^k\frac1{d^2}\left(\frac{d^2}s\right)^{2j}
 +\E\sbra*{T^{2\alpha}}
 \right]\right)\notag\\
 &\le O_\alpha\!\left(F^2\left[
 \frac{d^2}{s^2}
 +\left(\frac{d^2}s\right)^{2\alpha}
 \right]\right),
 \label{gt:eq:radial-final}
\end{align}
where \cref{gt:eq:radial-final} uses \cref{gt:lem:beta-sphere}, $d^2/s\leq 1$ and $k$ is a constant.
Since $s\ge d^2$,
$d^2/s^2\le1/s\le d^{1-1/\alpha}/s$.  Combining
\cref{gt:eq:total-var}, \cref{gt:eq:angular-final}, \cref{gt:eq:radial-final} proves
\cref{gt:eq:one-batch-var} and completes the proof of \cref{gt:prop:one-batch}.

\subsection{Technical lemmas}\label{gt:sec:technical-lemmas}

\subsubsection{Hayashi's POVM}

\begin{lemma}\label{gt:lem:beta-sphere}
Let $\ket{v}$ be the outcome of the Hayashi's POVM applied on $\ket{\psi}^{\otimes s}$. Choose the phase of the outcome so that $\braket{v}{\psi}\ge0$.  Then
$\ket v=\sqrt{1-T}\,\ket\psi+\sqrt T\,\ket G$,
where $G$ is uniform on the unit sphere of $\psi^\perp$,
$T$ and $G$ are independent, and
$T\sim\operatorname{Beta}(d^2-1,s+1)$.
For every fixed $q>0$ and $s\ge d^2$,
\begin{equation}\label{gt:eq:beta-real-moment}
 \E\sbra*{T^q}
 =\frac{\Gamma(d^2-1+q)\Gamma(d^2+s)}
        {\Gamma(d^2-1)\Gamma(d^2+s+q)}
 \le O_q\!\left(\left(\frac{d^2}{s}\right)^q\right).
\end{equation}
For integer $j\ge0$, $\E\sbra{T^j}=\mu_j(s)$, with $\mu_j$ as in
\cref{gt:eq:mu-def}.
\end{lemma}

\begin{proof}
For a Haar random $\ket{v}$, one has $|\braket{v}{\psi}|^2 \sim\operatorname{Beta}(1,d^2-1)$, with density function $f(x)=(d^2-1)(1-x)^{d^2-2}$.
These facts follow, for example, by writing a
Haar-random vector as a normalized standard complex Gaussian vector.

Now, let $\ket{v}$ be the outcome of Hayashi's POVM. Then $\ket{v}$'s density is obtained from Haar measure by multiplying by $\binom{d^2+s-1}{s}|\braket{v}{\psi}|^{2s}$. 
Therefore, the $|\braket{v}{\psi}|^2$ has density
\begin{align*}
 \binom{d^2+s-1}{s}(d^2-1)x^s(1-x)^{d^2-2}
 &=
 \frac{x^s(1-x)^{d^2-2}}{B(s+1,d^2-1)},
\end{align*}
where $B(\cdot,\cdot)$ denotes the Euler beta function.
Therefore
$|\braket{v}{\psi}|^2\sim\operatorname{Beta}(s+1,d^2-1)$.
This means
$T=1-|\braket{v}{\psi}|^2\sim\operatorname{Beta}(d^2-1,s+1),$
and
$\ket v=\sqrt{1-T}\,\ket\psi+\sqrt T\,\ket G$,
where $G$ is on the unit sphere of $\psi^\perp$. Furthermore, one can easily check that the law of $G$ conditioned on any value of $T$ is unitary invariant on $\psi^\perp$, which means $G$ is uniformly distributed and independent of $T$.

The beta-integral identity
\cite[Eq.~(5.12.1)]{NISTDLMF} gives
\begin{align*}
 \E\sbra*{T^q}
 &=
 \frac{1}{B(d^2-1,s+1)}
 \int_0^1 t^{d^2-2+q}(1-t)^s\,\mathrm{d}t\\
 &=
 \frac{B(d^2-1+q,s+1)}{B(d^2-1,s+1)}\\
 &=
 \frac{\Gamma(d^2-1+q)\Gamma(d^2+s)}
      {\Gamma(d^2-1)\Gamma(d^2+s+q)}.
\end{align*}
To prove the required bound, first let $r\ge1$ be an integer. Then
\begin{align*}
 \E\sbra*{T^r}
 &=
 \prod_{k=0}^{r-1}\frac{d^2-1+k}{d^2+s+k}\le
 \left(\frac{d^2+r}{s}\right)^r
 \le
 (r+1)^r\left(\frac{d^2}{s}\right)^r
 =O_r\!\left(\left(\frac{d^2}{s}\right)^r\right).
\end{align*}
For general $q>0$, write $q=m+\theta$, where
$m=\lfloor q\rfloor$ and $0\le\theta<1$. By H\"older's inequality,
\[
 \E\sbra*{T^q}
 \le
 \rbra*{\E\sbra*{T^m}}^{1-\theta}
 \rbra*{\E\sbra*{T^{m+1}}}^\theta
 \le
 O_q\!\left(\left(\frac{d^2}{s}\right)^q\right).
\]
This proves \cref{gt:eq:beta-real-moment}.

Finally, for an integer $j\ge1$,
\[
 \E\sbra*{T^j}
 =
 \prod_{r=0}^{j-1}
 \frac{d^2-1+r}{d^2+s+r}
 =\mu_j(s),
\]
and the case $j=0$ is immediate.
\end{proof}

\subsubsection{Coefficients}

\begin{lemma}\label{gt:lem:richardson}
The system \cref{gt:eq:rich-system} has a unique solution, and
\begin{equation}\label{gt:eq:a-bound}
 \sum_{\ell=0}^k|a_\ell|+\sum_{\ell=0}^k a_\ell^2\le O_\alpha(1).
\end{equation}
\end{lemma}

\begin{proof}
The factors
\[
 \prod_{r=0}^{j-1}(d^2-1+r)
\]
in \cref{gt:eq:mu-def} depend only on \(j\) and are nonzero.  It
therefore suffices to consider the matrix
\[
 B_{j\ell}
 :=
 \left[\prod_{r=0}^{j-1}(x_\ell+r)\right]^{-1},
 \qquad
 x_\ell:=d^2+s_\ell,
\]
and \cref{gt:eq:rich-system} is thus equivalent to $B a = (1,0,\ldots,0)^{{\mathsf T}}$.

Now, we prove $B$ is non-singular. Multiply column \(\ell\) of $B$ by $\prod_{r=0}^{k-1}(x_\ell+r)$, then the entry in row \(j\) and column $\ell$ then becomes
\[
\prod_{r=j}^{k-1}(x_\ell+r),
\]
which is a monic polynomial in \(x_\ell\) of degree \(k-j\).
After reversing the row order, the rows are evaluations of monic
polynomials of degrees \(0,1,\ldots,k\). 
By elementary row operations (subtracting suitable linear combinations of the preceding rows), this matrix is transformed into the ordinary Vandermonde matrix
$
 \{x_\ell^j\}_{0\le j,\ell\le k}.
$
Thus its determinant is, up to the sign coming from reversing the rows,
$\prod_{0\le \ell<h\le k}(x_h-x_\ell)$.
Since \(x_0<\cdots<x_k\), this determinant is nonzero, and hence the
system has a unique solution.

For the uniform estimate, rescale row \(j\) of $B$ by \(m^j\) and set 
$q:=\frac{d^2}{m}\in(0,1]$, $\tau:=\frac1m\in(0,1/4]$.
The rescaled matrix $C$ has entries
\[
 C_{j\ell}
 =
 \prod_{r=0}^{j-1}\frac1{2^\ell+q+r\tau}.
\]
The same Vandermonde calculation gives
\[
 \abs{\det C}
 =
 \frac{\displaystyle
   \prod_{0\le\ell<h\le k}(2^h-2^\ell)}
 {\displaystyle
   \prod_{\ell=0}^k\prod_{r=0}^{k-1}
   (2^\ell+q+r\tau)}.
\]
For $(q,\tau) \in [0,1]\times[0,1/4]$,
the absolute values of entries of \(C\) are uniformly upper bounded, while the displayed
determinant is bounded below by a positive constant depending only on
\(k\).  Hence the absolute values of entries of \(C^{-1}\) are uniformly upper bounded.
Note that $a$ is also the solution of $Ca=(1,0,\ldots,0)^{\mathsf T}$ since the row rescaling $R$ does not change the right hand side, i.e., let $R$ denote the row rescaling matrix, then
\[
 Ca= RBa=R(1,0,\ldots,0)^{\mathsf T}=(1,0,\ldots,0)^{\mathsf T},
\]
each \(a_\ell\) is bounded by a constant depending only on \(k\), and
therefore only on \(\alpha\).  As the number of coefficients is fixed,
\cref{gt:eq:a-bound} follows.
\end{proof}

\subsubsection{Schatten norm expansion and moments}

\begin{lemma}[{\cite{PS14}}]\label{gt:lem:schatten-taylor}
Fix $p>2$.  There are constants $C_p$ and
$C_{p,j}$ such that, for every $d\ge1$ and every $M\in\C^{d\times d}$ with
$\|M\|_p\le1$, there are symmetric real multilinear forms $L_{j,M}:(\C^{d\times d})^j\to\mathbb R$ for $1\le j\le \lceil p\rceil-1$,
where $\C^{d\times d}$ is viewed as a real vector space, satisfying
\begin{equation}\label{gt:eq:derivative-bound}
 \abs*{L_{j,M}(H_1,\ldots,H_j)}
 \le C_{p,j}\|M\|_p^{p-j}\prod_{i=1}^j\|H_i\|_p
\end{equation}
for any $H_1,\ldots,H_j\in\C^{d\times d}$, and for every
$H\in\C^{d\times d}$ with $\|H\|_p\le1$, it holds that
\begin{equation}\label{gt:eq:Taylor-uniform}
 \|M+H\|_p^p
 =\|M\|_p^p
 +\sum_{j=1}^{\lceil p\rceil-1} L_{j,M}(H,\ldots,H)
 +\mathcal R_p(M,H), \qquad 
 |\mathcal R_p(M,H)|
 \le C_p\|H\|_p^p.
\end{equation}
\end{lemma}

\begin{proof}
Let $r = \ceil{p} - 1$. 
Since $r<p\le r+1$, \cite[Theorem~16]{PS14} provides symmetric real
multilinear forms $\delta_A^{(j)}$ for $1\le j\le r$, which, by \cite[Equations (23)--(24) and Theorem 18]{PS14}, satisfy
\[
 |\delta_A^{(j)}(V_1,\ldots,V_j)|
 \le C_{p,j}\|A\|_p^{p-j}\prod_{i=1}^j\|V_i\|_p.
\]
Although \cite[Theorem~16]{PS14} states its remainder only as $\|V\|_p\to0$, the final
step of its proof gives
\[
 \left|\|A_1\|_p^p-\|A_0\|_p^p
 -\sum_{j=1}^r\delta_{A_0}^{(j)}(V,\ldots,V)\right|
 \le C_p\|V\|_p^p,
 \qquad V=A_1-A_0,
\]
whenever $\|A_0\|_p,\|A_1\|_p\le1$; see
\cite[Equation~(29), Theorem~14, Remark~15, and pp.~471--473]{PS14}.

For the matrices in the lemma, take $A_0=M/2$, $V=H/2$, and
$A_1=(M+H)/2$, which satisfy $\Abs{A_0}_p \leq 1$ and $\Abs{A_1}_p \leq 1$.
The proof is completed by setting
\[
 L_{j,M}(H_1,\ldots,H_j)
 :=2^{p-j}\delta_{M/2}^{(j)}(H_1,\ldots,H_j).
\]
\end{proof}

\begin{lemma}\label{gt:lem:isotropic}
Let $M\in\C^{d\times d}$ satisfy $\|M\|_{\mathrm F}=1$, and let $G$ be uniform on the unit
sphere of the complex Frobenius hyperplane $M^\perp = \cbra{ X \in \C^{d \times d} \mid \langle M,X\rangle_{\mathrm F} = 0 }$.  For every fixed $q\ge2$ and $u>0$,
\begin{equation}\label{gt:eq:isotropic-bound}
 \left(\E\sbra*{\|G\|_q^u}\right)^{1/u}
 \le O_{q,u}\rbra{d^{1/q-1/2}}.
\end{equation}
In particular,
\begin{equation}\label{gt:eq:G-p}
 \E\sbra*{\|G\|_{2\alpha}^{2\alpha}}
 \le O_\alpha \rbra*{d^{1-\alpha}}
 \le O_\alpha\rbra*{\mathrm{F}_\alpha(\rho)}
\end{equation}
whenever $\rho=MM^\dag$.
\end{lemma}

\begin{proof}
Let $W\in\C^{d\times d}$
have independent standard complex Gaussian entries, i.e., $W_{ij} \sim \mathcal N_{\C}(0,1)$, and set $Z\coloneqq W-\langle M,W\rangle_{\mathrm F}M$.
Then $Z$ is a standard complex Gaussian vector in $M^\perp$.  Thus, in any complex
orthonormal basis $E_1,\ldots,E_{d^2-1}$ of $M^\perp$,
\[
 Z=\sum_{j=1}^{d^2-1}\zeta_jE_j,
 \qquad
 \zeta_1,\ldots,\zeta_{d^2-1}\sim
 \mathcal N_{\C}(0,1).
\]
Set $\widetilde G:=Z/\Abs{Z}_{\mathrm F}$. 
As noted in \cite[proof of Proposition~6.34, p.~171]{AS17}, 
$\widetilde G$ is uniform on the unit sphere of $M^\perp$ and is independent of $\Abs{Z}_{\mathrm F}$;
in particular, $\widetilde G$ has the same distribution as $G$. 
Therefore,
\[
 \E\sbra*{\|Z\|_q^u}
 =\E\sbra*{\Abs{Z}_{\mathrm F}^u\|\widetilde G\|_q^u}
 =\E\sbra*{\Abs{Z}_{\mathrm F}^u} \E\sbra*{\|G\|_q^u},
\]
and hence
\begin{equation}\label{gt:eq:isotropic-moment-ratio}
 (\E\sbra*{\|G\|_q^u})^{1/u}
 =\frac{(\E\sbra*{\|Z\|_q^u})^{1/u}}{(\E\sbra*{\Abs{Z}_{\mathrm F}^u})^{1/u}}.
\end{equation}
Note that $|\zeta_j|^2 \sim \operatorname{Exp}\rbra{1}$ (cf.\ \cite[Equation~(A.2), p.~307]{AS17}), and thus
\[
 \Abs{Z}_{\mathrm F}^2=\sum_{j=1}^{d^2-1}|\zeta_j|^2\sim\operatorname{Gamma}(d^2-1,1).
\]

Now that $q\ge2$ and $\|M\|_q\le\|M\|_{\mathrm F}=1$, we have
\begin{align} \label{eq:Zneq}
 \|Z\|_q
 \leq \Abs{W}_q + \abs*{\langle M,W\rangle_{\mathrm F}} \Abs*{ M }_q \le d^{1/q}\|W\|_{\mathrm{op}}
      +|\langle M,W\rangle_{\mathrm F}|.
\end{align}
For $u>0$, put $K_u:=2^{\max\{u-1,0\}}$, so that
$(x+y)^u\le K_u(x^u+y^u)$ for all $x,y\ge0$.  
By \cite[Proposition~6.33, Equation~(6.42), p.~169]{AS17}, we have
\[
 \Pr \sbra*{ \|W\|_{\mathrm{op}}>2\sqrt d+t }
 \le\frac12e^{-t^2},\qquad t>0.
\]
Let $Y:=\max\{\|W\|_{\mathrm{op}}-2\sqrt d,0\}$.  Then,
\begin{align*}
 \E \sbra*{Y^u}
 &=\int_0^\infty ut^{u-1}\Pr\sbra{Y>t}\,\mathrm{d}t\\
 &\le\frac u2\int_0^\infty t^{u-1}e^{-t^2}\,\mathrm{d}t
 =\frac12\Gamma\!\left(1+\frac u2\right).
\end{align*}
Since $\|W\|_{\mathrm{op}}\le2\sqrt d+Y$, it follows that
\begin{equation} \label{eq:EWuop}
 \E\sbra*{\|W\|_{\mathrm{op}}^u}
 \le K_u\left(2^ud^{u/2}+\frac12\Gamma\!\left(1+\frac u2\right)\right)
 \le O_u \rbra*{ d^{u/2} },
\end{equation}
and hence $(\E\sbra{\|W\|_{\mathrm{op}}^u})^{1/u}\le O_u\rbra{\sqrt d}$.

Next, note that
$\langle M,W\rangle_{\mathrm F} \sim\mathcal N_{\C}(0,1)$.
Therefore,
\cite[Equation~(A.2), p.~307]{AS17} gives
\[
 \E\sbra*{|\langle M,W\rangle_{\mathrm F}|^u}=\Gamma\!\left(1+\frac u2\right).
\]
By \cref{eq:Zneq,eq:EWuop}, we have
\begin{align*}
 \E\sbra*{\|Z\|_q^u}
 \le K_u\left(
      d^{u/q}\E\sbra*{\|W\|_{\mathrm{op}}^u}+\E\sbra*{|\langle M,W\rangle_{\mathrm F}|^u}
     \right)
 \le O_{q,u} \rbra*{d^{u(1/q+1/2)}},
\end{align*}
and hence $(\E\sbra{\|Z\|_q^u})^{1/u}\le O_{q,u}\rbra{d^{1/q+1/2}}$.
Since $M^\perp$ has complex dimension $d^2-1$, the gamma-moment formula gives
\[
 (\E\sbra*{\Abs{Z}_{\mathrm F}^u})^{1/u}
 =\left(
   \frac{\Gamma(d^2-1+u/2)}{\Gamma(d^2-1)}
  \right)^{1/u}
 = \Theta_u\rbra{d}.
\]
Substituting these two estimates into
\cref{gt:eq:isotropic-moment-ratio} yields
\[
 \rbra*{\E\sbra*{\|G\|_q^u}}^{1/u}
 \le O_{q,u}\rbra*{d^{1/q-1/2}}.
\]
\end{proof}

\subsubsection{Some facts}

\begin{lemma}\label{gt:lem:endpoint-scalar}
For every $\beta>0$, there is a polynomial $J_\beta$ of degree at most
$\lceil\beta\rceil-1$ such that
\[
 |(1-t)^\beta-J_\beta(t)|\le O_\beta (t^\beta),
 \qquad(0\le t\le1).
\]
One may take the Taylor polynomial at $t=0$ through degree
$\lceil\beta\rceil-1$.
\end{lemma}

\begin{proof}
For $0\le t\le1/2$, Taylor's theorem gives a remainder of order
$t^{\lceil\beta\rceil}$, which is at most $t^\beta$.  For $1/2\le t\le1$, both the function and
the fixed polynomial are bounded, while $t^\beta\ge2^{-\beta}$.
\end{proof}

\begin{lemma}\label{gt:lem:beta-variance}
Let $T\sim\operatorname{Beta}(d^2-1,s+1)$, where $d\ge2$ and $s\ge d^2$.  For every fixed
integer $j\ge1$,
\begin{equation}\label{gt:eq:beta-var}
 \Var\sbra{T^j}
 \le O_j\!\left(\frac{1}{d^2}\left(\frac{d^2}{s}\right)^{2j}\right).
\end{equation}
\end{lemma}

\begin{proof}
Set $a=d^2-1$, $b=s+1$, and $M_j:=\E\sbra{T^j}$.  The beta moment formula gives
\[
 \frac{M_{2j}}{M_j^2}
 =\prod_{r=0}^{j-1}
   \frac{a+j+r}{a+r}\frac{a+b+r}{a+b+j+r}
 \le\prod_{r=0}^{j-1}\left(1+\frac{j}{a+r}\right)
 \le1+O_j\!\left(\frac{1}{d^2}\right).
\]
Thus $\Var\sbra{T^j}\le O_j(M_j^2/d^2)$.  Also
$M_j\le O_j((d^2/s)^j)$ by \cref{gt:eq:beta-real-moment}, proving the claim.
\end{proof}

\begin{lemma}\label{gt:lem:relative-gradient}
For every state $\rho$ and every $\alpha>1$,
\begin{equation}\label{gt:eq:spectral-relative}
 \tr(\rho^{2\alpha-1})
 \le\mathrm{F}_\alpha(\rho)^{2-1/\alpha}
 \le\mathrm{F}_\alpha(\rho)^2d^{1-1/\alpha}.
\end{equation}
\end{lemma}

\begin{proof}
Let $\lambda_{\max}$ be the largest eigenvalue of $\rho$.  Then
\[
 \tr(\rho^{2\alpha-1})
 \le\lambda_{\max}^{\alpha-1}\mathrm{F}_\alpha(\rho)
 \le\mathrm{F}_\alpha(\rho)^{(\alpha-1)/\alpha}
      \mathrm{F}_\alpha(\rho),
\]
because $\lambda_{\max}^\alpha\le\mathrm{F}_\alpha(\rho)$.  The second inequality in
\cref{gt:eq:spectral-relative} follows from
$\mathrm{F}_\alpha(\rho)\ge d^{1-\alpha}$.
\end{proof}

\begin{proof}
The assumption gives
$\widehat F/F\ge e^{-\gamma}$ and
$\widehat F/F\le2-e^{-\gamma}$.  Since
$e^\gamma+e^{-\gamma}\ge2$, one has
$2-e^{-\gamma}\le e^\gamma$.  Thus
$e^{-\gamma}\le\widehat F/F\le e^\gamma$, and taking logarithms proves the claim.
\end{proof}

\begin{lemma}\label{gt:lem:spherical-poincare}
Let $H$ be a real inner-product space of dimension $m\geq 2$, and let
\[
 \mathbb{S}(H):=\{x\in H:\|x\|=1\}.
\]
If $G$ is uniformly distributed on $\mathbb{S}(H)$, then every
continuously differentiable function
$f:\mathbb{S}(H)\to\mathbb{R}$ satisfies
\begin{equation}
\label{gt:eq:spherical-poincare}
 \Var\sbra{f(G)}
 \leq
 \frac{1}{m-1}
 \E\sbra*{\left\|
 \nabla_{\mathbb{S}(H)}f(G)
 \right\|^2}.
\end{equation}
Here $\nabla_{\mathbb{S}(H)}f$ denotes the tangential gradient of $f$ along the sphere.

In particular, if $\widetilde f$ is a continuously differentiable
extension of $f$ to a neighborhood of $\mathbb{S}(H)$ in $H$, then
\begin{equation}
\label{gt:eq:spherical-poincare-ambient}
 \Var\sbra{f(G)}
 \leq
 \frac{1}{m-1}
 \E\sbra*{\left\|
 \nabla \widetilde f(G)
 \right\|^2},
\end{equation}
where $\nabla$ is simply the Euclidean gradient with respect to the inner product on $H$.
\end{lemma}

\begin{proof}
After identifying $H$ isometrically with $\mathbb{R}^m$,
\cref{gt:eq:spherical-poincare} is the classical Poincar\'e
inequality on $\mathbb{S}^{m-1}$, applied to
$f-\E\sbra{f(G)}$; see
\cite[Proposition~2.3]{BianchiniColesantiPagniniRoncoroni2023}.

For the second assertion, observe that
$\nabla_{\mathbb{S}(H)}f(x)$ is the orthogonal projection of
$\nabla_H\widetilde f(x)$ onto the tangent space
$T_x\mathbb{S}(H)$.  Hence
\[
 \left\|\nabla_{\mathbb{S}(H)}f(x)\right\|
 \leq
 \left\|\nabla\widetilde f(x)\right\|,
\]
and \cref{gt:eq:spherical-poincare-ambient} follows from
\cref{gt:eq:spherical-poincare}.
\end{proof}

\section*{Acknowledgment}

The authors used Large Language Models as AI-assisted research and writing tools throughout the
preparation of this paper. The tools were used to help brainstorm ideas and explore proof
strategies. Portions of the paper text were redrafted or modified with AI assistance across all
sections. All final mathematical claims, algorithms, proofs, citations, and wording were reviewed,
edited, and validated by the authors. The authors assume responsibility for all content of the paper.

\addcontentsline{toc}{section}{References}

\bibliographystyle{alphaurl}
\bibliography{main}

\end{document}